%% file: main.tex
\documentclass[11pt, onecolumn]{article}

\usepackage{todonotes}
\usepackage{amsmath, amssymb, amsthm}
\usepackage{mathrsfs} 
\usepackage{geometry}
\usepackage[utf8]{inputenc}
\usepackage{hyperref}
\usepackage{graphicx}
\usepackage{tikz}\usepackage{tikz-cd}
\usepackage{braket}
\usetikzlibrary{patterns, positioning, shapes.geometric}

\usepackage{mathtools}  

\usepackage{bm}

\newtheorem{definition}{Definition}
\newtheorem{theorem}{Theorem}
\newtheorem{lemma}{Lemma}
\newtheorem{corollary}{Corollary}
\newtheorem{remark}{Remark}
\newtheorem{example}{Example}
\newtheorem{proposition}{Proposition}
\newtheorem*{lemmaA}{Lemma A}
\newtheorem*{corollaryA}{Corollary A}

\providecommand{\keywords}[1]{\par\medskip\noindent\textbf{Keywords:} #1\par}
\providecommand{\subclass}[1]{\par\noindent\textbf{Mathematics Subject Classification (2020):} #1\par}

\newcommand{\cH}{\mathcal{H}}

\newcommand{\cD}{\mathcal{D}}

\usepackage[T1]{fontenc}
\usepackage[utf8]{inputenc}
\usepackage[T1]{fontenc}
\usepackage[utf8]{inputenc}
\usepackage{amsmath}
\usepackage[all]{xy}
\usepackage{booktabs}   
\usepackage{graphicx}   
\usepackage{algorithm}      
\usepackage{algpseudocode}  

\usepackage[all]{xy}
\usepackage{amssymb} 
\usepackage{newunicodechar}
\newunicodechar{⋯}{\cdots}

\newcommand{\remove}[1]{}

\usepackage{tabularx}

\title{Finite-Precision Quantum Mechanics: Quantum Parcels, Information and Geometry }
\author{
    Abbas Edalat\\
    Imperial College London, UK\\
    \texttt{a.edalat@imperial.ac.uk}\\
    ORCID: \href{https://orcid.org/0000-0002-6211-1991}{0000-0002-6211-1991}}
\begin{document}\maketitle

\begin{abstract}
Quantum mechanics represents states by exactly specified density operators,
whereas experimentally available information is necessarily finite in
precision. We develop \emph{Interval Quantum Mechanics (IQM)}, in which
finite observational information is represented by a \emph{quantum parcel},
a convex open set of density operators. Experimental parcels are determined
by finitely many expectation intervals and form a basis for the state-space
topology. A single parcel represents the states compatible with the available
information, while a double parcel $(O_1,O_2)$ consists of disjoint possible
and excluded regions. Exact point states are recovered as ideal limits of
successive parcel refinement.

Unitary dynamics lifts to a reversible evolution of parcels. Finite-resolution
measurements are represented by fuzzy POVMs and Kraus updates; single parcels
are preserved, while double-parcel updates are order-compatible whenever they
remain double parcels. We give a useful sufficient condition for such
preservation, together with examples showing it is not necessary. Under
suitable conditions measurement contracts ambient Hilbert--Schmidt volume and
strictly increases the associated geometric information. This contrasts with
von Neumann entropy, which can remain unchanged under a selective measurement
even though information has been gained.

In this formulation, several familiar foundational paradoxes no longer arise
in standard form. Wave--particle duality becomes a continuous suppression of
interference as which-path resolution increases. Schr\"odinger's cat is
described by a finite parcel updated toward the observed outcome sector,
rather than an exact superposition undergoing abrupt collapse. Entanglement
remains genuinely nonclassical: CHSH violation persists throughout an open
parcel around a Bell state.


\subclass{81P05 \textbullet\ 81P15 \textbullet\ 81P40 \textbullet\ 47L07}
\end{abstract}

\maketitle

\input{introduction}
\input{quantum-parcels}
\input{reduction-theorems}

\input{unitary-evolution}
\input{measurement}
\input{Uncertainty-and-Entropy}
\input{particle-wave}
\input{cat-paradox}
\input{entanglement}
\input{computational-aspects}

\input{conclusion}

\input{appendix}

\section*{Statements and Declarations}

\subsection*{Funding}
No funding was received to assist with the preparation of this manuscript.

\subsection*{Competing Interests}
The author has no competing interests to declare that are relevant to the
content of this article.

\subsection*{Data Availability}
This is a theoretical study. No datasets were generated or analysed during
the current study.

\bibliographystyle{plainurl}
\bibliography{refs}
\end{document}

%% file: introduction.tex
\section{Introduction}

Standard quantum mechanics is formulated in terms of idealised exact objects:
point states, exact real-valued probabilities, and sharp measurements. Physical
experiments, by contrast, have finite resolution and provide only finitely many
constraints on the state of a system. This distinction is particularly
pronounced for macroscopic systems, for which only a small number of
macroscopic observables, such as total energy, magnetisation, or particle
number, are experimentally accessible. The available data therefore do not
operationally determine a unique microscopic density matrix.

The tension between exact mathematical description and operational
accessibility has long been recognised. Heisenberg's original formulation of
quantum mechanics was motivated by the principle that the theory should be
expressed in terms of observable quantities~\cite{Heisenberg1925}. Von Neumann
likewise emphasised the limited relevance of complete microscopic
specification for macroscopic systems~\cite{vonNeumann1955}. More recently,
the work of Goldstein, Lebowitz, Tumulka, and Zanghì~\cite{GLTZ2010} on normal
typicality and of Reimann~\cite{Reimann2015} on quantum thermalisation has
shown that important macroscopic predictions can be largely insensitive to
the precise microscopic state.

In this paper we develop a finite-precision reformulation of quantum
mechanics that retains the standard Hilbert-space formalism but changes the
representation of experimentally accessible state information. A finite
record of strict expectation-value constraints determines a convex weak open
set of compatible density matrices, which we call a \emph{quantum parcel}.
The physically fundamental examples are \emph{experimental parcels}, defined
by finitely many such constraints; they form a basis for the state-space
topology. A single parcel records the states compatible with the available
information. A \emph{double-parcel} $(O_1,O_2)$ additionally records an
excluded region: $O_1$ contains the possible states and $O_2$ states ruled
out by the observational history. Reverse inclusion of possible regions,
together with inclusion of excluded regions, gives the natural information
order. We call the resulting framework \emph{Interval Quantum Mechanics}
(IQM).

Point states are not removed from the mathematical theory. Rather, IQM
separates exact states from finite observational descriptions and establishes
their precise relation by reduction. Experimental parcels can be successively
refined so that a nested sequence converges to any prescribed density
operator, in both finite and infinite dimensions. Correspondingly, the
prediction ranges of bounded observables converge to their standard
point-state values; for non-negative unbounded observables the analogous
result holds under a uniform moment condition. Standard point-state quantum
mechanics is therefore recovered as the infinite-precision limit of the
parcel formulation.

This relation between IQM and standard quantum mechanics is analogous to the
relation between finite numerical representations and exact real-number
analysis. Just as a floating-point or fixed-precision numeral represents not
a single real number but the range of reals consistent with its finite
specification, while the real continuum itself remains the exact ideal
theory approximated to arbitrary accuracy by refining the representation, a
quantum parcel represents not a single density operator but the range of
states consistent with a given finite observational record, while the
point-state Hilbert-space formalism remains the exact ideal theory recovered
in the limit of arbitrarily fine observational refinement
(Theorems~\ref{thm:reduction-finite} and~\ref{thm:reduction-infinite-bounded}
below). In both settings the finite-precision object is not a deficient
substitute for the exact one: it is the mathematically precise
representation of what is actually specified at any given stage, whether of
computation or of observation. This is the sense in which IQM should be
understood throughout the paper: not as a modification of quantum mechanics,
but as a reformulation of what it means to specify a quantum state with
finite information.

Unitary dynamics lifts directly from states to parcels. A unitary $U$ sends
a parcel $O$ to $UOU^\dagger$ and acts componentwise on double-parcels.
This evolution is deterministic, reversible and preserves the information
order. In finite dimensions it also preserves ambient Hilbert--Schmidt
volume, while operational volume is preserved when the observational
subspace is transported with the dynamics. Thus finite observational
descriptions evolve without introducing any modification of the underlying
unitary quantum dynamics.

Measurement, by contrast, changes the information represented by a parcel.
Finite-resolution measurements are modelled by fuzzy POVMs and their Kraus
updates. For non-sharp measurements the associated Kraus map is a
homeomorphism of state space and sends every single parcel to another parcel.
For double-parcels, the possible region is updated by the observed outcome,
while the excluded region retains previously excluded states together with
their updated images. We characterize when this update remains a
double-parcel and give a useful sufficient condition for preservation.
Under the hypotheses established in Section~\ref{sec:measurement},
sufficiently sharp measurements contract ambient Hilbert--Schmidt volume,
leading to a strict increase of the corresponding geometric information for
single parcels and, whenever the double-parcel structure is preserved, for
double-parcels. This information gain need not be reflected in the von
Neumann entropy: for example, a selective sharp measurement of a pure
superposition can produce a pure outcome state, so that the entropy remains
zero before and after measurement while the measurement has nevertheless
supplied information.

The usual scalar predictions of quantum mechanics acquire range-valued
counterparts at finite precision. Expectation values and Born probabilities
are obtained by ranging the corresponding affine functionals over the
parcel, while standard deviations likewise become uncertainty ranges. Von
Neumann entropy becomes an entropy range describing the mixedness of the
compatible states. These quantities are distinct from the geometric
information associated with the size of possible and excluded regions; the
latter measures refinement of the compatible state region rather than the
mixedness of its constituent states.

An important consequence of this finite-precision reformulation is that
several familiar foundational paradoxes no longer arise in their standard
form. Rather than requiring separate interpretational mechanisms,
wave--particle duality and Schr\"odinger's cat acquire geometric
descriptions in terms of parcels and their measurement-induced refinement,
while Bell nonclassicality remains robust throughout suitable open parcels
of entangled states. Sections~\ref{sec:double-slit}--\ref{sec:entanglement}
develop these geometric formulations explicitly.
\subsection{Relation to finite-information and foundational approaches}

The general claim that infinite precision need not be regarded as physically
fundamental has important precedents. Of particular relevance is the work of
Del Santo and Gisin~\cite{DelSantoGisin2025}. They identify a
\emph{principle of infinite precision} implicit in the conventional
interpretation of classical mechanics and replace infinitely specified real
quantities by \emph{finite-information quantities} (FIQs), whose
undetermined components are represented by propensities. The resulting
indeterminacy is ontic, and they show that several phenomena often associated
with quantum mechanics have analogues in an indeterministic interpretation
of classical physics.

IQM shares with this programme the motivation to take finite information
seriously, but its mathematical construction and purpose are different.
IQM retains the ordinary quantum state space $\cD(\cH)$ together with
non-commuting observables, Born probabilities, unitary dynamics, POVMs and
Kraus updates. Finite precision enters through the information available
about a quantum state: finitely many observational constraints determine a
convex weak open region of quantum state space. The problem is therefore not
to obtain indeterminism by modifying classical states, but to develop the
topology, dynamics, measurement theory and information structure of
finite-resolution quantum state descriptions.

The distinction is particularly clear for phenomena arising from
incompatible observables. Del Santo and Gisin emphasise that their
finite-information classical theory does not reproduce genuinely quantum
effects based on incompatibility; in particular, its analogue of
entanglement cannot violate Bell inequalities.

IQM, by contrast, operates inside the usual non-commutative quantum
structure: a sufficiently small parcel around an entangled state can itself
certify entanglement, and can be chosen so that every density matrix in the
parcel violates a specified Bell inequality (Theorem~\ref{bell-violation}).
Thus finite precision is common to the two approaches, while the underlying
state spaces and resulting theories are different.

IQM is also related to epistemic, operational and information-order
approaches to quantum theory. Helland~\cite{Helland2024} reconstructs
quantum structure from accessible variables and symmetry principles, whereas
IQM starts with the conventional Hilbert-space formalism and asks how its
states and dynamics should be represented when observational information has
finite resolution. Geometric logic and domain theory regard open sets and
refinement orders as natural representations of finite information
\cite{johnstone1982stone,vic89,abramsky1991domain,scott1970outline,
smyth1983,abramsky_domain_1994}; this motivates the use of weak open parcels
and reverse inclusion in IQM. Coecke and Martin~\cite{coecke2010partial}
instead order individual density matrices through their spectral information
content. Interval probabilities also occur in imprecise probability theory
\cite{walley1991statistical} and in earlier quantum settings
\cite{hartle1968quantum}; in IQM such intervals arise from ranging ordinary
Born probabilities over a finite-resolution state parcel.

The framework differs likewise from subjective epistemic interpretations
such as QBism~\cite{fuchs2010,fuchs2014introduction}: an experimental parcel
is fixed by a specified and shareable collection of observational
constraints rather than by an individual agent's degree of belief.
Generalised probabilistic theories~\cite{chiribella2010probabilistic,
Plavala2023} and topos-theoretic approaches
\cite{butterfield1999topos,doring2008topos,heunen2009topos} provide much
broader reconstructions or reformulations of quantum theory; IQM instead
retains ordinary quantum mechanics and studies the set-valued state
descriptions generated by finite observational information.

A different route to finite-resolution physics is taken by
Treppiedi~\cite{Treppiedi2026}, who derives the Hilbert-space formalism of
quantum mechanics, the Born rule, and the kinematic and dynamical structure
of quantum field theory from a single algebraic symmetry principle: finite-resolution
representations of the additive group $(\mathbb R,+)$ subject to four
consistency constraints. This is a representation-theoretic reconstruction
of quantum theory itself from finite resolution, extending to quantum field
theory and gravity; IQM instead retains the standard Hilbert-space formalism
throughout and constructs finite-precision state descriptions directly
within it.

IQM suggests an experimentally accessible proof of concept. In finite-sample
quantum tomography, finite-sample observations naturally determine a
statistically certified non-singleton region of state space compatible with
the measured data, whereas a single reconstructed density matrix requires an
additional inferential choice. By repeating the experiment with increasing
sample size, one can study the refinement of these certified regions and
compare their interval predictions with subsequent measurements. Such an
experiment cannot rule out an underlying exact density matrix, but it tests
the operational premise of IQM: that the finite-resolution object directly
warranted by experimental evidence is a parcel rather than a point state.

The parcel construction is not confined to the foundational examples studied
in the present paper. Subsequent work has extended it in two directions.
In~\cite{EdalatErgodicity2026}, Reimann's quantum ergodic and thermalisation
framework is formulated directly for quantum parcels: under an appropriate
effective-dimension condition, expectation intervals become concentrated
around equilibrium values for most late times. In
\cite{EdalatIAQFT2026}, the construction is extended to algebraic quantum
field theory, where finite local observations determine weak$^*$ open convex
parcels of states. The resulting Interval Algebraic Quantum Field Theory
develops parcel reduction, measurement and locality and reformulates several
operator-algebraic structures of AQFT within the finite-precision framework.

These developments indicate that quantum parcels are not specific to
finite-dimensional quantum tomography or to the foundational examples
considered below. Rather, they provide a common representation of finite
observational information that can be carried from ordinary quantum
mechanics to many-body thermalisation and relativistic quantum field theory.

%

\subsection{Main results}

The main mathematical results of the present paper are as follows:
\begin{itemize}

\item Experimental parcels generated by finitely many expectation intervals
form a basis for the weak topology in the density-operator state space.
Standard point-state descriptions are recovered under nested refinement,
both in finite dimensions and, under appropriate hypotheses, in infinite
dimensions, including for unbounded observables subject to uniform moment
bounds.

\item Unitary evolution lifts to a deterministic and reversible evolution of
parcels and, in finite dimensions, preserves ambient Hilbert--Schmidt volume,
while expectation values and Born probabilities become interval-valued.

\item Finite-resolution measurements are represented by fuzzy POVMs and
their Kraus updates. Single parcels are preserved, while double-parcel
updates are order-compatible whenever they remain double parcels; we give a
useful sufficient condition for preservation and examples showing that it
is not necessary.

\item For qubits and finite multi-qubit systems we establish contraction of
ambient Hilbert--Schmidt volume under suitable measurement updates, yielding
an associated geometric information measure that increases under the
conditions established below.

\item We apply the framework to wave--particle duality, Schr\"odinger's cat
and entanglement, and develop hyper-rectangular observational
representations of state-space outer approximations for computational use.

\end{itemize}

The remainder of the paper develops these results in turn, beginning with
the parcel formalism and its reduction to standard point-state quantum
mechanics, followed by dynamics, measurement, foundational applications,
geometric information, and computational representations.

\subsection{Technical Background}
\label{sec:technical-background}

We assume familiarity with the standard formalism of quantum mechanics:
Hilbert spaces, density matrices, observables, and positive
operator-valued measures (POVMs); see, for example,
\cite{Bengtsson2017,Conway1990,Watrous2018}. Let $\mathcal H$ be a
separable complex Hilbert space, let $\mathcal T(\mathcal H)$ denote the
trace-class operators, and let $\mathcal B(\mathcal H)$ denote the bounded
operators. These spaces are in duality under
\[
\langle T,B\rangle=\operatorname{Tr}(TB),
\qquad
\mathcal T(\mathcal H)^*=\mathcal B(\mathcal H).
\]
The state space is
\[
\mathcal D(\mathcal H)
=
\{\rho\in\mathcal T(\mathcal H):
\rho\geq0,\ \operatorname{Tr}\rho=1\}.
\]

We equip $\mathcal D(\mathcal H)$ with the relative weak topology
$\sigma(\mathcal T(\mathcal H),\mathcal B(\mathcal H))$. This is the
Banach-space weak topology on $\mathcal T(\mathcal H)$, not the weak*
topology on $\mathcal B(\mathcal H)$. By definition it is the coarsest
topology making every functional
$\rho\mapsto\operatorname{Tr}(\rho B)$, $B\in\mathcal B(\mathcal H)$,
continuous; see, for example, \cite[Ch.~V]{Conway1990}.

Since density matrices are self-adjoint, Hermitian operators suffice to
generate this topology. Indeed, every $B\in\mathcal B(\mathcal H)$ has
the decomposition
\[
B=A+iC,\qquad
A=\frac{B+B^*}{2},\qquad
C=\frac{B-B^*}{2i},
\]
with $A$ and $C$ Hermitian, and
\[
\operatorname{Re}\operatorname{Tr}(\rho B)=\operatorname{Tr}(\rho A),
\qquad
\operatorname{Im}\operatorname{Tr}(\rho B)=\operatorname{Tr}(\rho C).
\]
Thus the expectation cylinders
\[
\{\rho\in\mathcal D(\mathcal H):
a<\operatorname{Tr}(\rho H)<b\},
\qquad H=H^*,
\]
form a subbasis for the relative weak topology, and their finite
intersections form a basis. They are convex because the expectation
functionals are linear.

These expectation functionals also \emph{separate points}: if
$\rho\neq\sigma$, there exists a bounded Hermitian observable $H$ such
that
$\operatorname{Tr}(\rho H)\neq\operatorname{Tr}(\sigma H)$.
For example, one may take $H=\rho-\sigma$, since
\[
\operatorname{Tr}(\rho H)-\operatorname{Tr}(\sigma H)
=
\operatorname{Tr}((\rho-\sigma)^2)>0.
\]

On the state space $\mathcal D(\mathcal H)$, this weak topology in fact
coincides with the trace-norm topology, where the trace norm of
$T\in\mathcal T(\mathcal H)$ is
\[
\|T\|_1:=\operatorname{Tr}|T|
=\operatorname{Tr}\sqrt{T^\dagger T}.
\]
We first show that the weak topology is coarser. For every
$B\in\mathcal B(\mathcal H)$ and
$\rho,\sigma\in\mathcal D(\mathcal H)$,
\[
\bigl|\operatorname{Tr}((\rho-\sigma)B)\bigr|
\leq
\|\rho-\sigma\|_1\,\|B\|.
\]
Thus every evaluation map
$\rho\mapsto\operatorname{Tr}(\rho B)$ is trace-norm continuous. Since
$\sigma(\mathcal T(\mathcal H),\mathcal B(\mathcal H))$ is by definition
the coarsest topology making all these evaluation maps continuous, the
weak topology is coarser than the trace-norm topology.

For the
converse, let $\rho_\alpha\to\rho$ be any weakly convergent net in
$\mathcal D(\mathcal H)$; then $\rho_\alpha\to\rho$ in the weak operator
topology, with $\operatorname{Tr}\rho_\alpha=\operatorname{Tr}\rho=1$
fixed throughout. Gr\"umm's convergence theorem
\cite[Thm.~2.19]{Simon2005} (originally~\cite{Grumm1973}) therefore gives
$\|\rho_\alpha-\rho\|_1\to0$. Hence every trace-norm-open set is also
weak-open, and the two topologies coincide on $\mathcal D(\mathcal H)$.

Since $\mathcal H$ is separable, so is $\mathcal T(\mathcal H)$ in trace
norm, so $\mathcal D(\mathcal H)$ is a separable metric space in the
trace norm and therefore second countable
\cite[Thm.~30.3]{Munkres2000}. By the coincidence just established, the
relative weak topology on $\mathcal D(\mathcal H)$ is second countable
as well, hence first countable; consequently its convergence and closure
properties may be described using sequences rather than nets throughout
the remainder of the paper.

In finite dimensions this topology is the usual Euclidean topology and
$\mathcal D(\mathcal H)$ is compact. In infinite dimensions
$\mathcal D(\mathcal H)$ is not weakly compact in general. Accordingly,
compactness arguments below are confined to finite dimensions, while
Theorems~\ref{thm:reduction-infinite-bounded}
and~\ref{thm:unbounded-moment} use explicit trace-norm concentration and
moment bounds.

%% file: quantum-parcels.tex
\section{Quantum Parcels}
\label{sec:parcels}

Consider a quantum experiment prepared and characterized by finitely many
observations. Each observation is subject to the finite precision inherent
in any physical apparatus: it does not determine the corresponding physical
quantity exactly, but locates its value within an open interval consistent
with the resolution of the measurement. Such observations therefore do not
single out a density matrix as the state of the system. Rather, they determine
the set of all density matrices whose expectation values for the observed
quantities lie within the corresponding intervals. It is this set, rather
than any one of its constituent density matrices, that is warranted by the
finite experimental information. We formalize this operational state
description as a \emph{quantum parcel}.

\begin{definition}[Experimental parcel]
\label{def:experimental-parcel}
An \emph{experimental parcel} is a finite intersection of expectation cylinders,
\[
O = \{\rho \in \mathcal{D}(\mathcal{H}) : a_j < \operatorname{Tr}(\rho H_j) < b_j,\; j=1,\dots,m\},
\quad H_j = H_j^*\in\mathcal B(\mathcal H),
\]
for arbitrary bounded $H_1,\dots,H_m$ and finite or infinite $\dim\mathcal H$. Experimental parcels are convex weak open sets and are precisely the parcels directly certifiable by a finite set of coarse-grained measurements.
\end{definition}
Experimental parcels have a simple convex-geometric description. An
experimental parcel as defined above is an \emph{open spectrahedral
region}: it is obtained by intersecting the positive trace-one state
space with finitely many open affine strips determined by the
observational constraints. This description applies in both finite and
infinite dimensions. In finite dimensions, moreover, the closure of an
experimental parcel is a compact convex spectrahedron; in infinite
dimensions closures need not be compact without further assumptions.

Associated with $O$ is its \emph{state-space outer approximation},
obtained by dropping positivity while retaining the trace-one,
self-adjoint and observational constraints:
\[
O^{\mathrm{outer}}
=
\left\{
X\in\mathcal T(\mathcal H)_{\mathrm{sa}}:
\operatorname{Tr}X=1,\quad
a_j<\operatorname{Tr}(XH_j)<b_j,\quad j=1,\ldots,m
\right\}.
\]
Thus $O\subseteq O^{\mathrm{outer}}$.  This outer approximation retains
exactly the finite observational constraints defining $O$ but allows
trace-one self-adjoint operators that need not be positive. In infinite dimensions $O^{\mathrm{outer}}$ is necessarily unbounded
whenever it is non-empty. Indeed, finitely many observational
functionals together with the trace constraint leave a non-zero
infinite-dimensional common kernel in
$\mathcal T(\mathcal H)_{\mathrm{sa}}$; hence, starting from any
$X_0\in O^{\mathrm{outer}}$, one may move arbitrarily far along a
non-zero direction in this kernel without changing any of the defining
constraints. In finite dimensions, by contrast, the outer approximation
can be bounded when the observational directions span the full trace-one
affine space. Its
hyper-rectangular representation in finite observational coordinates is
developed in Section~\ref{sec:computational}.

Convexity is physically natural — if two density matrices are both consistent with the measured data, so is any mixture of them — and unitary evolution (Theorem~\ref{thm:unitary-evolution}) preserves it, as does a single fuzzy measurement update (Proposition~\ref{prop:singleupdate}), which maps a convex weak open set to a convex weak open set for every fuzziness $\eta\in(0,1)$; neither operation, in general, preserves the \emph{experimental} form. This motivates the more general notion:

\begin{definition}[Quantum parcel]
\label{def:single-parcel}
A \emph{quantum parcel}, or simply \emph{parcel}, is a non-empty convex weak open subset $O\subseteq\mathcal D(\mathcal H)$. Every experimental parcel is a parcel; general parcels arise from experimental ones under unitary evolution and fuzzy measurement.
\end{definition}

We say $O'$ is \emph{more informative} than $O$, written $O\sqsubseteq O'$, if $O'\subseteq O$.

A single parcel records states compatible with the available information but does not separately record states that have been explicitly excluded. This motivates the following extension.

\begin{definition}[Double-parcel]
A \emph{double-parcel} is a pair $(O_1,O_2)$ with $O_1,O_2\subseteq\mathcal D(\mathcal H)$ weak open and convex, $O_1\neq\emptyset$, and $O_1\cap O_2=\emptyset$. Here $O_1$ is the set of \emph{possible} states and $O_2$ the set of \emph{impossible} states.
\end{definition}
Disjointness follows automatically from the measurement apparatus: if $O_1$ is certified by a reading of some observable $A$ and $O_2$ by a macroscopically distinct, non-overlapping bin, then $O_1\cap O_2=\emptyset$ by construction, and $O_2\neq\emptyset$ whenever that bin meets the spectrum of $A$. The natural order on double-parcels is
\[
(O_1,O_2)\sqsubseteq(O_1',O_2') \iff O_1'\subseteq O_1 \text{ and } O_2\subseteq O_2',
\]
i.e.\ more informative means a smaller possible set and a larger impossible set. A single parcel cannot express that a state has been ruled out — it can only shrink (more knowledge) or grow (forgetting) — so the double-parcel is the minimal structure tracking both possible and excluded states; the measurement updates studied in Section~\ref{sec:measurement} preserve this order whenever both updated components remain genuine parcels (Proposition~\ref{prop:monotone}), which holds whenever the two updated components — not merely convex and open, which they are automatically — are also disjoint.

\subsection{Operational volume and geometric information}

Let
$\mathcal{HS}(\mathcal H)
=
\{A\in\mathcal B(\mathcal H):
\operatorname{Tr}(A^\dagger A)<\infty\}$
denote the space of Hilbert--Schmidt operators on $\mathcal H$, equipped
with the Hilbert--Schmidt inner product
$\langle A,B\rangle_{\mathrm{HS}}
=\operatorname{Tr}(A^\dagger B)$.
In finite dimensions
$\mathcal{HS}(\mathcal H)=\mathcal B(\mathcal H)$, whereas in infinite
dimensions $\mathcal{HS}(\mathcal H)$ is a proper subspace of
$\mathcal B(\mathcal H)$; for example,
$I\notin\mathcal{HS}(\mathcal H)$.

We first define a volume associated with the finitely many observables
used to specify an experimental parcel. Let $O$ be an experimental
parcel defined by bounded Hermitian Hilbert--Schmidt observables
$H_1,\ldots,H_m$ and intervals $(a_j,b_j)$.

Suppose first that $d=\dim\mathcal H<\infty$. Adding a multiple of the
identity to an observable merely shifts all its expectation values by
the same constant, since $\operatorname{Tr}(\rho I)=1$ for every
$\rho\in\mathcal D(\mathcal H)$. Thus the part of $H_j$ that distinguishes
states is its traceless part
\[
H_j^0
=
H_j-\frac{\operatorname{Tr}H_j}{d}I.
\]
We therefore define the observational subspace by
\[
V
=
\operatorname{span}_{\mathbb R}
\{H_1^0,\ldots,H_m^0\}
\qquad
(\dim\mathcal H<\infty).
\]

If $\dim\mathcal H=\infty$, the identity operator is not
Hilbert--Schmidt, so there is no identity direction within
$\mathcal{HS}(\mathcal H)$ to remove. In this case we define
\[
V
=
\operatorname{span}_{\mathbb R}
\{H_1,\ldots,H_m\}
\qquad
(\dim\mathcal H=\infty).
\]
In either case $V$ is a finite-dimensional real Hilbert space under the
Hilbert--Schmidt inner product. Write $m'=\dim V\leq m$ and choose a
Hilbert--Schmidt-orthonormal basis
$\Psi_1,\ldots,\Psi_{m'}$ of $V$.
The corresponding \emph{observational coordinate map} is
\[
\Phi_V(\rho)
=
\bigl(
\operatorname{Tr}(\rho\Psi_1),\ldots,
\operatorname{Tr}(\rho\Psi_{m'})
\bigr)
\in\mathbb R^{m'}.
\]
It records the component of the state visible through the chosen
finite-dimensional family of observables. Let $\lambda_{m'}$ denote
Lebesgue measure on $\mathbb R^{m'}$. We define the
\emph{operational volume} of $O$ relative to $V$ by
\[
\operatorname{Vol}_{\mathrm{op},V}(O)
:=
\lambda_{m'}\bigl(\Phi_V(O)\bigr).
\]
This definition is independent of the chosen
Hilbert--Schmidt-orthonormal basis of $V$, since two such bases are
related by an orthogonal transformation, which preserves Lebesgue
measure.

The image $\Phi_V(O)$ should be distinguished from the coordinate box
specified by the observational intervals. Positivity of density
matrices imposes additional constraints. In particular, when the
defining observables themselves form the chosen
Hilbert--Schmidt-orthonormal basis of $V$,
\[
\Phi_V(O)
=
\Phi_V(\mathcal D(\mathcal H))
\cap
\prod_{j=1}^{m'}(a_j,b_j),
\]
which may be a proper subset of the corresponding coordinate box.

In finite dimensions there is a second, generally different, notion of
volume. Let $d=\dim\mathcal H$ and take $V$ to be the full
$(d^2-1)$-dimensional real space of traceless Hermitian operators. Then
$\Phi_V$ is an affine isometric embedding of the state space, since
\[
\|\Phi_V(\rho)-\Phi_V(\sigma)\|_2
=
\|\rho-\sigma\|_{\mathrm{HS}}
\qquad
(\rho,\sigma\in\mathcal D(\mathcal H)).
\]
For this full choice of $V$ we write
\[
\operatorname{Vol}_{\mathrm{HS}}(S)
:=
\operatorname{Vol}_{\mathrm{op},V}(S)
\]
and call it the \emph{ambient Hilbert--Schmidt volume}. Thus
$\operatorname{Vol}_{\mathrm{HS}}(O)$ measures the volume of the whole
parcel in the full state space, whereas
$\operatorname{Vol}_{\mathrm{op},V}(O)$ measures the volume of its
lower-dimensional observational shadow. The two notions coincide when
$V$ is the full ambient space.

For example, consider a qubit. Its state space is the
three-dimensional Bloch ball. If the only observational direction is
the $z$ direction, an experimental parcel specified by
\[
a<\operatorname{Tr}(\rho\sigma_z)<b
\]
is a slab of the Bloch ball. Its observational shadow is the interval
obtained by projecting this slab onto the $z$ direction, so its
operational volume is the one-dimensional length of that interval
(up to the normalization associated with the
Hilbert--Schmidt-normalized Pauli observable). Its ambient
Hilbert--Schmidt volume, by contrast, is the three-dimensional volume
of the entire slab. Thus the two volumes quantify different aspects of
the same parcel: the variation visible through the chosen observations
and the size of the full region of compatible quantum states,
respectively. If all three Pauli directions are observed, then $V$ is
the full three-dimensional traceless Hermitian space and the two
notions of volume coincide.

\begin{remark}[Infinite dimensions]
\label{rem:volume-finite-infinite}
When $\dim\mathcal H=\infty$, the state space
$\mathcal D(\mathcal H)$ has no finite-dimensional ambient volume, and
a parcel defined by finitely many expectation constraints is generally
not contained in a finite-dimensional affine subspace. There is
therefore no analogue of $\operatorname{Vol}_{\mathrm{HS}}$ for the
whole state space. Nevertheless, if the finitely many observational
directions span a finite-dimensional Hilbert--Schmidt subspace $V$,
the shadow $\Phi_V(O)\subset\mathbb R^{\dim V}$ still has the ordinary
finite-dimensional Lebesgue volume
$\operatorname{Vol}_{\mathrm{op},V}(O)$. This volume is finite because
for every bounded observable $H$ and density matrix $\rho$,
$|\operatorname{Tr}(\rho H)|\leq\|H\|$.
\end{remark}

Corresponding to these two notions of volume are two geometric
information functionals. In finite dimensions we define the
\emph{ambient information} by
\[
I_{\mathrm{HS}}(O)
:=
\frac{1}{\operatorname{Vol}_{\mathrm{HS}}(O)},
\qquad
I_{\mathrm{HS}}(O_1,O_2)
:=
\frac{\operatorname{Vol}_{\mathrm{HS}}(O_2)}
{\operatorname{Vol}_{\mathrm{HS}}(O_1)}.
\]
These are the quantities used in the volume-contraction results of
Section~\ref{sec:measurement}, where sufficient conditions for their
increase under measurement are established.

For any finite-dimensional Hilbert--Schmidt observational subspace $V$,
we analogously define the \emph{observational information}
$I_{\mathrm{op},V}(O)$ and
$I_{\mathrm{op},V}(O_1,O_2)$ by replacing
$\operatorname{Vol}_{\mathrm{HS}}$ in the preceding formulas by
$\operatorname{Vol}_{\mathrm{op},V}$. These quantities remain meaningful
for suitable parcels in infinite-dimensional Hilbert spaces. In finite dimensions the observational and ambient
notions coincide when $V$ is the full space of traceless Hermitian
operators.

\paragraph{Extreme points.}
For convex $C\subseteq\mathcal T(\mathcal H)$ with weak closure $\overline C$, $\rho\in\overline C$ is an \emph{extreme point} of $\overline C$ if it is not a proper convex combination of two distinct points of $\overline C$. When $\overline C$ is weak compact, Krein--Milman gives $\overline C$ as the closed convex hull of its extreme points (a finite convex hull in finite dimensions, by Carathéodory), every weakly continuous linear functional attains its extrema on $\overline C$ at such points, and for a bounded polyhedron the extreme points are its vertices. Examples of weak compact sets in finite dimensions include $\mathcal D(\mathcal H)$ and the closure of any experimental parcel.

\paragraph{Openness and pure states.}
Any open neighbourhood containing a pure state necessarily also contains mixed states arbitrarily close to it. Thus finite observational localisation cannot certify purity exactly: purity appears only as a limiting point-state property under successive refinement.

\begin{proposition}[No open set of exclusively pure states]\label{prop:no-open-pure-states}
Let $\mathcal{H}$ be a Hilbert space of dimension $\ge 2$, and let $\mathcal{D}(\mathcal{H})$ be the space of density matrices endowed with the weak topology. No non-empty weak open set $O \subset \mathcal{D}(\mathcal{H})$ can consist exclusively of pure states.
\end{proposition}

%% file: reduction-theorems.tex
\section{Reduction to Standard Quantum Mechanics}\label{sec:reduction}

In IQM, exact point-state descriptions are recovered under successive refinement of finite observational information. We first formulate finite-resolution prediction ranges and then establish the reduction in finite and infinite dimensions.

In standard quantum mechanics, an observable \(H\) yields a definite expectation value \(\operatorname{Tr}(\rho H)\) for a given state \(\rho\). When the state is a quantum parcel (a weak open convex set), this single value generalises to a range of possible values.

\begin{definition}[Expectation range]
\label{def:expectation}
Let $O \subset \mathcal{D}(\mathcal{H})$ be a quantum parcel and $H$ a Hermitian operator. The \emph{expectation range} of $H$ over $O$ is
\[
E_O(H) := \{\operatorname{Tr}(\rho H) : \rho \in O\} \subseteq \mathbb{R}.
\]
Since $O$ is convex and the expectation functional is affine, $E_O(H)$ is an interval (Proposition~\ref{prop:expectation}); depending on endpoint attainment, it may be open, closed, half-open, or a singleton. We record the associated closed enclosure
\[
\overline E_O(H) := \Bigl[\inf_{\rho\in O}\operatorname{Tr}(\rho H),\ \sup_{\rho\in O}\operatorname{Tr}(\rho H)\Bigr],
\]
which bounds every predicted value regardless of endpoint attainment. The same observations apply verbatim to probability ranges in Definition~\ref{def:probability}, since $\rho\mapsto\operatorname{Tr}(\rho E)$ is likewise affine in $\rho$ for a POVM element $E$.
\end{definition}

\begin{proposition}[Expectation ranges are intervals]
\label{prop:expectation}
Let $O\subseteq\mathcal D(\mathcal H)$ be non-empty and convex, and let $H$ be Hermitian. Then $E_O(H)=\{\operatorname{Tr}(\rho H):\rho\in O\}$ is an interval with endpoints
\[
\alpha=\inf_{\rho\in O}\operatorname{Tr}(\rho H),
\qquad
\beta=\sup_{\rho\in O}\operatorname{Tr}(\rho H);
\]
every value strictly between $\alpha$ and $\beta$ is attained by some $\rho\in O$. Whether either endpoint itself is attained depends on $O$ and $H$.
\end{proposition}

\begin{theorem}[Reduction to point states: finite dimensions]
\label{thm:reduction-finite}
Let $\dim\mathcal{H} = N < \infty$ and let $\rho_0 \in \mathcal{D}(\mathcal{H})$. Let $\Psi_1,\dots,\Psi_{N^2-1}$ be a Hilbert--Schmidt-orthonormal basis of the traceless Hermitian operators of Section~\ref{sec:parcels}, so that $\Phi_V$ (for $V$ the full ambient space) is an affine isometric embedding of $\mathcal D(\mathcal H)$. Define the nested sequence of parcels
\[
O_n(\rho_0) = \Bigl\{\rho \in \mathcal{D}(\mathcal{H}) :
|\operatorname{Tr}(\rho \Psi_j) - \operatorname{Tr}(\rho_0 \Psi_j)| < 2^{-n},\;
j = 1,\dots,N^2-1\Bigr\}.
\]
Then $\rho_0 \in O_n(\rho_0)$ for all $n$, $\bigcap_{n=1}^\infty O_n(\rho_0) = \{\rho_0\}$, and $\operatorname{diam}_1(O_n(\rho_0))\to0$. Moreover, for every bounded Hermitian operator $H$ on $\mathcal{H}$,
\[
\lim_{n\to\infty} \inf_{\rho \in O_n(\rho_0)} \operatorname{Tr}(\rho H)
= \lim_{n\to\infty} \sup_{\rho \in O_n(\rho_0)} \operatorname{Tr}(\rho H)
= \operatorname{Tr}(\rho_0 H).
\]
\end{theorem}

\begin{remark}[Outer approximations collapse in finite dimensions]
\label{rem:finite-outer-collapse}
The outer set $O_n^{\mathrm{outer}}(\rho_0)$, obtained by dropping the positivity constraint from $O_n(\rho_0)$ above, satisfies the identical bound $\operatorname{diam}_1(O_n^{\mathrm{outer}}(\rho_0))\to0$ by the same computation, since the $\Psi_j$ already span the entire ambient trace-one hyperplane; hence it too collapses to $\{\rho_0\}$. This is special to the finite-dimensional case, where finitely many observables can span the whole state space, and does not extend to infinite dimensions.
\end{remark}

\begin{theorem}[Reduction to point states: infinite dimensions]
\label{thm:reduction-infinite-bounded}
Let $\dim\mathcal H=\infty$ and $\rho_0\in\mathcal D(\mathcal H)$. There exists a nested sequence of experimental parcels
\[
O_1(\rho_0)\supseteq O_2(\rho_0)\supseteq\cdots,
\qquad
\rho_0\in O_n(\rho_0)\subseteq B_1(\rho_0,2^{-n}):=\bigl\{\rho\in\mathcal D(\mathcal H):\|\rho-\rho_0\|_1<2^{-n}\bigr\},
\]
for every $n$. Consequently $\bigcap_n O_n(\rho_0)=\{\rho_0\}$, and for every bounded Hermitian $H$,
\[
\sup_{\rho\in O_n(\rho_0)}\bigl|\operatorname{Tr}(\rho H)-\operatorname{Tr}(\rho_0H)\bigr| \le \|H\|\,2^{-n}\xrightarrow[n\to\infty]{}0,
\]
so the expectation ranges collapse to the standard value:
\[
\operatorname{diam} E_{O_n(\rho_0)}(H)\longrightarrow0,
\qquad
\inf E_{O_n(\rho_0)}(H),\ \sup E_{O_n(\rho_0)}(H)\ \longrightarrow\ \operatorname{Tr}(\rho_0H).
\]
\end{theorem}

\begin{remark}
\label{rem:no-infinite-outer}
Unlike the finite-dimensional case, this argument does not extend to an outer set obtained by dropping positivity from finitely many linear constraints on $\mathcal T(\mathcal H)$: finitely many such constraints leave an infinite-dimensional kernel when $\dim\mathcal H=\infty$, so no analogue of Remark~\ref{rem:finite-outer-collapse} is available. The reduction result is therefore stated here directly for genuine experimental parcels.
\end{remark}

\begin{theorem}[Reduction for unbounded observables under a uniform moment bound]
\label{thm:unbounded-moment}
Let $\{O_n(\rho_0)\}_{n\in\mathbb N}$ be the nested experimental parcels of
Theorem~\ref{thm:reduction-infinite-bounded}, and let $H\ge 0$ be a
self-adjoint, possibly unbounded operator such that
\[
\operatorname{Tr}(\rho_0H)<\infty.
\]
Assume that there exists $\delta>0$ such that
\[
\sup_{n\in\mathbb N}\sup_{\rho\in O_n(\rho_0)}
\operatorname{Tr}(\rho H^{1+\delta})<\infty.
\]
Then
\[
\lim_{n\to\infty}
\sup_{\rho\in O_n(\rho_0)}
\left|
\operatorname{Tr}(\rho H)-\operatorname{Tr}(\rho_0H)
\right|=0.
\]
Equivalently,
\[
\inf E_{O_n(\rho_0)}(H),\;
\sup E_{O_n(\rho_0)}(H)
\longrightarrow
\operatorname{Tr}(\rho_0H).
\]
\end{theorem}

\begin{remark}
The additional moment assumption in Theorem~\ref{thm:unbounded-moment}
is a uniform integrability condition controlling the high-energy tails of the
states in the parcels. It is satisfied for many physically important
unbounded observables on natural classes of states, including the harmonic
oscillator Hamiltonian
\[
H=\frac12(P^2+X^2)=a^\dagger a+\frac12,
\]
the number operator \(N=a^\dagger a\), and more general Schr\"odinger
operators of the form
\[
H=-\Delta+V(x)
\]
with confining potentials \(V(x)\). For thermal states, coherent states,
Gaussian states, and more generally for families of states with uniformly
bounded energy moments, one typically has
\[
\sup_n\sup_{\rho\in O_n}\operatorname{Tr}(\rho H^{1+\delta})<\infty.
\]

By contrast, even trace-norm convergence does not control expectations of
unbounded observables. Let \(H e_k=\lambda_k e_k\), with
\(\lambda_k\to\infty\), and define
\[
\rho_k=
\left(1-\frac1k\right)\rho_0
+
\frac1k|e_k\rangle\langle e_k|.
\]
Then
\[
\|\rho_k-\rho_0\|_1
\le \frac{2}{k}\to0,
\]
so $\rho_k\to\rho_0$ in trace norm, equivalently weakly (Section~\ref{sec:parcels}). Whereas
\[
\operatorname{Tr}(\rho_kH)
=
\left(1-\frac1k\right)\operatorname{Tr}(\rho_0H)
+
\frac{\lambda_k}{k},
\]
which diverges whenever $\lambda_k/k\to\infty$. Thus an additional
uniform moment condition is necessary in general even under trace-norm
convergence.
\end{remark}

\begin{corollary}[Standard QM as the infinite-precision limit]
\label{cor:limit}
Standard point-state quantum mechanics is recovered as the infinite-precision limit of interval quantum mechanics, under arbitrarily fine refinement of the defining observational constraints (Theorems~\ref{thm:reduction-finite} and~\ref{thm:reduction-infinite-bounded}); in finite dimensions this also holds for outer polyhedral approximations (Remark~\ref{rem:finite-outer-collapse}). For bounded observables the convergence is unconditional. For non-negative self-adjoint observables, the same conclusion holds under the uniform moment bound of Theorem~\ref{thm:unbounded-moment}.
\end{corollary}

At every finite stage of refinement, non-degenerate parcels generally yield non-singleton prediction ranges. The diameters of these prediction ranges quantify the residual uncertainty associated with the finite observational specification and vanish in the point-state limit established above.

%% file: unitary-evolution.tex
\section{Interval Unitary Dynamics}\label{sec:unitary}

Standard unitary evolution lifts naturally from point states to quantum parcels. If $U(t)$ is the unitary propagator, the point-state map $\rho\mapsto U(t)\rho U(t)^\dagger$ induces a reversible evolution of parcels by direct image.

\begin{definition}[Unitary evolution of a parcel]
\label{def:unitary}
Assume a unitary operator \(U\) on \(\mathcal{H}\). For a single quantum parcel \(O\subset\mathcal{D}(\mathcal{H})\) and a double quantum parcel with components \(O_1,O_2\subseteq\mathcal{D}(\mathcal{H})\), \(O_1\cap O_2=\emptyset\), define the evolved parcels, respectively, as
\[
U(O)=\{U\rho U^\dagger\mid\rho\in O\},\qquad U(O_1,O_2)=(U(O_1),U(O_2)).
\]
Since unitary conjugation is a linear weak homeomorphism of \(\mathcal D(\mathcal H)\), it preserves convexity, weak openness and disjointness; hence \(U(O)\) and \(U(O_1,O_2)\) are, respectively, a valid single and a valid double quantum parcel.
\end{definition}

\begin{theorem}[Properties of unitary evolution]
\label{thm:unitary-evolution}
Unitary evolution on single-parcels \(O,O'\) and on a double-parcel \((O_1,O_2)\) satisfies:
\begin{enumerate}
    \item \textbf{Reversibility:} \(U^{-1}\cdot(U\cdot(O))=O\) and \(U^{-1}\cdot(U\cdot(O_1,O_2))=(O_1,O_2)\).
    \item \textbf{Preservation of information order:} If \(O\sqsubseteq O'\) (i.e.\ \(O'\subseteq O\)), then \(U(O)\sqsubseteq U(O')\) (i.e.\ \(U(O')\subseteq U(O)\)). Similarly, if \((O_1,O_2)\sqsubseteq (O_1',O_2')\), then \(U\cdot(O_1,O_2)\sqsubseteq U\cdot(O_1',O_2')\).
    \item \textbf{Ambient volume preservation:} $\dim\mathcal H<\infty \implies \operatorname{Vol}_{\mathrm{HS}}(U(O)) = \operatorname{Vol}_{\mathrm{HS}}(O)$.
\end{enumerate}
Thus unitary evolution is deterministic, reversible, order-preserving, and ambient-volume-preserving.
\end{theorem}

\begin{remark}
Part (iii) concerns ambient Hilbert--Schmidt volume. Operational volume relative to a fixed observational subspace $V$ need not be preserved, since $U$ need not leave $V$ invariant: if $O$ is defined by $a_j<\operatorname{Tr}(\rho H_j)<b_j$, then for $\sigma=U\rho U^\dagger$ one has $\rho=U^\dagger\sigma U$ and $\operatorname{Tr}(\rho H_j)=\operatorname{Tr}(\sigma\,UH_jU^\dagger)$, so $U(O)$ is naturally described by the transported observables $UH_jU^\dagger$, i.e.\ $V'=UVU^\dagger$, not $V$ itself. If the observational subspace is transported this way, the operational volume is exactly preserved:
\[
\operatorname{Vol}_{\mathrm{op},V'}(U(O))=\operatorname{Vol}_{\mathrm{op},V}(O),
\qquad V'=UVU^\dagger,
\]
since, writing $\Psi_j'=U\Psi_jU^\dagger$ for an orthonormal basis $\{\Psi_j\}$ of $V$, we have $\operatorname{Tr}(U\rho U^\dagger\Psi_j')=\operatorname{Tr}(\rho\Psi_j)$ for every $\rho$: the observational coordinate images of $U(O)$ in $V'$ and of $O$ in $V$ are literally identical.
\end{remark}

A probability in quantum mechanics is a special case of an expectation value: the probability of outcome \(i\) in a measurement described by a POVM \(\{E_i\}\) is \(\operatorname{Tr}(\rho E_i)\), where each \(E_i\) is a positive operator (\(0\le E_i\le I\)) and \(\sum_i E_i=I\). In IQM, this becomes a probability range, following Section~\ref{sec:reduction}.

\begin{definition}[Probability range]
\label{def:probability}
Let \(O\subset\mathcal{D}(\mathcal{H})\) be a convex weak open set and \(E\) a POVM element (\(0\le E\le I\)). The \emph{probability range} for outcome \(E\) over \(O\) is
\[
\Pr_O(E) := \{\operatorname{Tr}(\rho E) : \rho\in O\}\subseteq[0,1],
\]
an interval by Proposition~\ref{prop:expectation}, with closed enclosure $\overline{\Pr}_O(E) := \bigl[\inf_{\rho\in O}\operatorname{Tr}(\rho E),\ \sup_{\rho\in O}\operatorname{Tr}(\rho E)\bigr]$.

For an experimental parcel $P$ whose state-space outer approximation
$P_{\mathrm{outer}}$, defined in Section~\ref{sec:parcels}, is a bounded
polyhedron with vertices $v_1,\ldots,v_k$, linearity of
$\rho\mapsto\operatorname{Tr}(\rho E)$ implies that its extrema over
$P_{\mathrm{outer}}$ occur at vertices, giving the conservative estimate
\[
\overline{\Pr}_{P_{\mathrm{outer}}}(E)
=
\left[
\min_i\operatorname{Tr}(v_iE),\,
\max_i\operatorname{Tr}(v_iE)
\right]
\supseteq
\overline{\Pr}_P(E).
\]
\end{definition}
\begin{example}[Qubit measurement]
\label{ex:qubit}
Consider a qubit parcel $P$ whose state-space outer approximation
$P_{\mathrm{outer}}$, as defined in Section~\ref{sec:parcels}, is a
bounded polyhedron. Its vertices need not themselves lie in
$\mathcal D(\mathbb C^2)$. For a projective measurement along the
$z$-axis with $E=|0\rangle\langle0|$, let $z_{\min}$ and $z_{\max}$
denote the minimum and maximum $z$-coordinates of the vertices of
$P_{\mathrm{outer}}$. Then
\[
\overline{\Pr}_P(E)
\subseteq
\Bigl[
\tfrac{1+z_{\min}}2,\,
\tfrac{1+z_{\max}}2
\Bigr]\cap[0,1].
\]
Indeed, $\overline{\Pr}_P(E)\subseteq[0,1]$, while
\[
\Bigl[
\tfrac{1+z_{\min}}2,\,
\tfrac{1+z_{\max}}2
\Bigr]
\]
is the outer bound obtained by optimizing over
$P_{\mathrm{outer}}$ and may extend beyond $[0,1]$ because
$P_{\mathrm{outer}}$ need not satisfy positivity and may therefore
extend outside the Bloch ball. The width of the resulting enclosure
quantifies our uncertainty about the outcome probability given our
incomplete knowledge of the state.
\end{example}

%% file: measurement.tex
\section{Measurement in Interval Quantum Mechanics}\label{sec:measurement}

In Interval Quantum Mechanics, measurements are modelled by fuzzy Positive Operator-Valued Measures (POVMs) with finite resolution. The fuzzy construction preserves the open-set character of single parcels and, under suitable conditions, the double-parcel structure. Assume we have a finite-dimensional Hilbert space $\mathcal H$, $\dim\mathcal H=d$, in this section.

Let \(\{\Pi_i\}_{i=1}^m\) be mutually orthogonal projectors defining a projective measurement with $m$ outcomes, satisfying
\[
\sum_{i=1}^m\Pi_i = I.
\]
To account for finite precision, we define a fuzzy POVM \(\{E_i\}_{i=1}^m\) as:
\[
E_i = \eta\Pi_i + \frac{1-\eta}{m}I,
\]
where \(0<\eta<1\) is the fuzziness parameter, representing measurement imperfection. Each \(E_i\) is strictly positive (\(E_i>0\)), so the associated Kraus operators are invertible.

The corresponding Kraus operators are \(M_i=\sqrt{E_i}\) (positive square root). These satisfy:
\[
\sum_{i=1}^m M_i^\dagger M_i = \sum_{i=1}^m E_i = \eta\sum_{i=1}^m\Pi_i+(1-\eta)I = I.
\]
Note that for \(\eta=1\), the projector \(E_i=\Pi_i\) is rank deficient while for \(\eta\in(0,1)\), it has full rank and is thus invertible.

\begin{definition}[Measurement update]
\label{def:update}

Let \(O\) be a single parcel and \((O_1,O_2)\) a double parcel. Since $\dim\mathcal H=d<\infty$, all parcels below lie in the same finite-dimensional trace-one affine hyperplane of $\mathcal T(\mathcal H)_{\mathrm{sa}}$.

For a fuzzy POVM measurement with outcome \(j\) and Kraus operator \(M_j^{(\eta)}\), the Kraus map is
\[
f_j^{(\eta)}(\rho)
=
\frac{M_j^{(\eta)} \rho (M_j^{(\eta)})^\dagger}
     {\operatorname{Tr}(\rho E_j^{(\eta)})},
\]
a homeomorphism of $\mathcal D(\mathcal H)$ for $\eta\in(0,1)$ (Proposition~\ref{prop:singleupdate}).

We define the updated parcels by
\[
O' = f_j^{(\eta)}(O),
\]
for a single parcel, and
\[
O_1' = f_j^{(\eta)}(O_1),
\qquad
O_2' = \operatorname{Conv}\bigl(O_2\cup f_j^{(\eta)}(O_2)\bigr),
\]
for a double parcel: the possible region updates via the observed outcome $j$ alone, and the excluded region retains everything previously excluded together with its image under that same, actually-observed update. We do not additionally incorporate the images $f_i^{(\eta)}(O_1)$, $i\neq j$, of the possible region under outcomes that did \emph{not} occur: these are counterfactual states that would have arisen on unrealized branches, not states excluded by the actually observed history, and including them would break monotonicity of the update (Proposition~\ref{prop:monotone} below).

Both $O_2$ and $f_j^{(\eta)}(O_2)$ are open, so their union is open, and the convex hull of an open set in finite dimension is open; hence $O_2'$ is open, and so is $O_1'$. Thus $O_1'$ and $O_2'$ are open convex sets; the updated single parcel $O'$ is a valid quantum parcel, and $(O_1',O_2')$ forms a double-parcel precisely when $O_1'\cap O_2'=\varnothing$.

\end{definition}

\begin{proposition}[Measurement update for a single-parcel]
\label{prop:singleupdate}
Let \(O\subset\mathcal{D}(\mathcal{H})\) be a single-parcel. Then, for all \(\eta\in(0,1)\), the updated single-parcel
\[
O' = f_j^{(\eta)}(O) = \left\{\frac{M_j^{(\eta)}\rho (M_j^{(\eta)})^\dagger}{\operatorname{Tr}(\rho E_j^{(\eta)})}:\rho\in O\right\}
\]
is a convex weak open set.
\end{proposition}

Thus, single-parcels are preserved under measurement.

\begin{proposition}[Monotonicity of the update formulas]
\label{prop:monotone}
For any measurement outcome \(j\) and any $\eta\in(0,1)$, the update formulas are monotone with respect to set inclusion: if $U_1\subseteq O_1$ and $O_2\subseteq U_2$, then
\[
U_1'\subseteq O_1', \qquad O_2'\subseteq U_2'.
\]
In particular, if \((O_1,O_2)\sqsubseteq (U_1,U_2)\) and both $(O_1',O_2')$ and $(U_1',U_2')$ are themselves double-parcels, then \((O_1',O_2')\sqsubseteq(U_1',U_2')\).
\end{proposition}

Thus the updated pair is a double-parcel precisely when
\[
f_j^{(\eta)}(O_1)\ \cap\ \operatorname{Conv}\bigl(O_2\cup f_j^{(\eta)}(O_2)\bigr)=\varnothing.
\]
The following theorem gives a useful sufficient condition for this disjointness to hold for sufficiently sharp measurements.

\begin{theorem}[A sufficient condition for preservation of double-parcels]
\label{thm:double-preservation}
Suppose there exist a Hermitian operator \(H=\Pi_jH\Pi_j\) and constants \(c_1>c_2\) such that, for all \(\rho_1\in\overline{O_1}\) and \(\rho_2\in\overline{O_2}\),
\[
\operatorname{Tr}(\rho_1H) > c_1\operatorname{Tr}(\rho_1\Pi_j),\qquad
\operatorname{Tr}(\rho_2H) < c_2\operatorname{Tr}(\rho_2\Pi_j).
\]
\emph{(Separation.)} Then this hypothesis alone already forces $\operatorname{Tr}(\rho\Pi_j)>0$ for every $\rho\in\overline{O_1}\cup\overline{O_2}$, hence (by compactness) a uniform bound $\delta:=\min_{\overline{O_1}\cup\overline{O_2}}\operatorname{Tr}(\rho\Pi_j)>0$; no separate positivity hypothesis is needed. There exists $\eta_0<1$ such that, for every $\eta\in(\eta_0,1)$,
\[
f_j^{(\eta)}(O_1)\ \cap\ \operatorname{Conv}\bigl(O_2\cup f_j^{(\eta)}(O_2)\bigr)=\varnothing,
\]
so $(O_1'(\eta),O_2'(\eta))$ is a double-parcel.
\end{theorem}

The ratio-form condition of Theorem~\ref{thm:double-preservation} is sufficient but not necessary: it is neither required for the double-parcel structure to survive measurement, nor is its failure enough to conclude that it does not. The following two elementary qubit examples establish both directions at once, and both rest on a single explicit computation.

\begin{lemma}[M\"obius map for the qubit fuzzy update]
\label{lem:mobius}
Let $\mathcal H=\mathbb C^2$, $m=2$, $\Pi_j=|0\rangle\langle0|$, and $p(\rho):=\operatorname{Tr}(\rho\Pi_j)$. For the two-outcome fuzzy update of Section~\ref{sec:measurement}, $p(f_j^{(\eta)}(\rho))$ depends only on $p(\rho)$, via
\[
F_\eta(p)=\frac{(1+\eta)p}{1-\eta+2\eta p},
\]
which is strictly increasing on $(0,1)$ with $F_\eta(p)>p$ there. Moreover $f_j^{(\eta)}$ maps every $p$-slab $\{\rho: a<p(\rho)<b\}$ bijectively onto $\{\rho: F_\eta(a)<p(\rho)<F_\eta(b)\}$.
\end{lemma}

\begin{example}[Preservation without separation]
\label{ex:preserved-no-separation}
Let $0<a<b<c<d<1$ and set
\[
O_2=\{\rho: a<p(\rho)<b\},\qquad O_1=\{\rho: c<p(\rho)<d\}.
\]
Since $\Pi_j$ has rank one, every $H=\Pi_jH\Pi_j$ is a real scalar multiple of $\Pi_j$, so $\operatorname{Tr}(\rho H)/\operatorname{Tr}(\rho\Pi_j)$ is constant wherever defined, and no $c_1>c_2$ can satisfy Theorem~\ref{thm:double-preservation}'s separation condition. Nonetheless, by Lemma~\ref{lem:mobius}, $f_j^{(\eta)}(O_1)=\{F_\eta(c)<p<F_\eta(d)\}$, and $O_2'(\eta)=\{a<p<F_\eta(b)\}$ — whether or not the two slabs $O_2$ and $f_j^{(\eta)}(O_2)$ themselves overlap, their convex hull is the full $p$-slab spanning the combined range, since $p$ is affine and convex combinations of points at the two extreme $z$-levels reach every intermediate $z$ and, at each such $z$, every point of the corresponding Bloch disk. Since $b<c$ gives $F_\eta(b)<F_\eta(c)$ for every $\eta\in(0,1)$, these $p$-ranges — and hence $O_1'(\eta)$ and $O_2'(\eta)$ themselves — are disjoint for \emph{every} $\eta\in(0,1)$: $(O_1'(\eta),O_2'(\eta))$ is a double-parcel throughout, despite the failure of the ratio-form condition.
\end{example}

\begin{example}[Failure without separation]
\label{ex:failure-no-separation}
With the same $\Pi_j$, $p$, and $0<a<b<c<d<1$, now set
\[
O_1=\{\rho: a<p(\rho)<b\},\qquad O_2=\{\rho: c<p(\rho)<d\}.
\]
As before, no ratio-form separation is possible. By Lemma~\ref{lem:mobius}, $F_\eta(a)\to1$ as $\eta\to1^-$, so $F_\eta(a)>c$ for $\eta$ sufficiently close to $1$. Since $a<b<d$ and $F_\eta$ is strictly increasing, $F_\eta(a)<F_\eta(b)<F_\eta(d)$; combined with $F_\eta(a)>c$, the nonempty interval $(F_\eta(a),F_\eta(b))$ satisfies
\[
(F_\eta(a),F_\eta(b))\ \subset\ (c,F_\eta(d)),
\]
so $p(O_1'(\eta))=(F_\eta(a),F_\eta(b))$ is entirely contained in $p(O_2'(\eta))=(c,F_\eta(d))$ for such $\eta$. Consequently $O_1'(\eta)\cap O_2'(\eta)\neq\varnothing$: $(O_1'(\eta),O_2'(\eta))$ fails to be a double-parcel for $\eta$ sufficiently close to $1$. Thus failure of the ratio-form condition alone does not determine whether the double-parcel structure is preserved: Examples~\ref{ex:preserved-no-separation} and~\ref{ex:failure-no-separation} exhibit, respectively, preservation and failure under the same absence of ratio-form separation.
\end{example}

\subsection{Volume contraction under measurements}

Consider a qubit with Hilbert space \(\mathcal{H}=\mathbb{C}^2\) and a two-outcome ($m=2$) fuzzy POVM with $\Pi_j=|0\rangle\langle0|$ (the projector onto the state \(|0\rangle\)). For a finite-precision measurement with sharpness parameter \(\eta\in(0,1]\), the POVM element is
\[
E_j^{(\eta)} = \eta\Pi_j + \frac{1-\eta}{2}I.
\]
The corresponding Kraus operator is \(M_j^{(\eta)}=\sqrt{E_j^{(\eta)}}\) and the update map for outcome \(j\) is
\[
f_j^{(\eta)}(\rho) = \frac{M_j^{(\eta)}\rho M_j^{(\eta)\dagger}}{\operatorname{Tr}(\rho E_j^{(\eta)})}.
\]
In the Bloch representation \(\rho=\frac12(I+x\sigma_x+y\sigma_y+z\sigma_z)\) the map reads
\[
x' = \frac{\sqrt{1-\eta^2}\,x}{1+\eta z},\qquad
y' = \frac{\sqrt{1-\eta^2}\,y}{1+\eta z},\qquad
z' = \frac{z+\eta}{1+\eta z}.
\]
Its Jacobian determinant (with respect to the volume on \(\mathbb{R}^3\)) is
\begin{equation}\label{jacob}
J^{(\eta)}(z) = \frac{(1-\eta^2)^2}{(1+\eta z)^4}.
\end{equation}

\begin{theorem}[Volume contraction for a qubit]
\label{thm:volcontraction_qubit}
Let \(P\subset\mathcal{D}(\mathbb{C}^2)\) be an open set of qubit states and assume that
\[
\operatorname{Tr}(\rho\Pi_j)\ge c>0 \qquad\text{for every }\rho\in P.
\]
Then there exists \(\eta_0\in(0,1)\) such that for every \(\eta\in(\eta_0,1)\),
\[
\operatorname{Vol}_{\mathrm{HS}}(f_j^{(\eta)}(P))<\operatorname{Vol}_{\mathrm{HS}}(P).
\]
If, moreover, \(c\ge\frac12\), then the inequality holds for all \(\eta\in(0,1)\).
\end{theorem}

\subsubsection*{General finite-dimensional case}
Now let \(\mathcal{H}\) be a $d$-dimensional Hilbert space and consider an $m$-outcome fuzzy POVM as in Section~\ref{sec:measurement} (the number of outcomes $m$ need not equal $d$). Let \(\Pi_j\) be a projector with rank \(r<d\) (the case \(r=d\) is trivial because then the measurement does not change the state). For \(\eta\in(0,1]\) define the POVM element:
\[
E_j^{(\eta)} = \eta\Pi_j + \frac{1-\eta}{m}I,
\]
and the corresponding update map
\[
f_j^{(\eta)}(\rho) = \frac{M_j^{(\eta)}\rho M_j^{(\eta)\dagger}}{\operatorname{Tr}(\rho E_j^{(\eta)})},\qquad M_j^{(\eta)}=\sqrt{E_j^{(\eta)}}.
\]

\begin{theorem}[Volume contraction, general finite-dimensional case]
\label{thm:volcontraction_nqubits}
Let \(P\subset\mathcal{D}(\mathcal{H})\) be a weak open subset of density matrices and assume
\[
\operatorname{Tr}(\rho\Pi_j)\ge c>0 \qquad\text{for every }\rho\in P.
\]
Then there exists \(\eta_0\in(0,1)\) such that for every \(\eta\in(\eta_0,1)\),
\[
\operatorname{Vol}_{\mathrm{HS}}(f_j^{(\eta)}(P))<\operatorname{Vol}_{\mathrm{HS}}(P).
\]
\end{theorem}

\begin{corollary}[Strict increase of geometric information]\label{cor:strict-increase}
Suppose $(O_1'(\eta),O_2'(\eta))$ is a double-parcel (in particular, this holds under the hypotheses of Theorem~\ref{thm:double-preservation}), that $O_2\neq\varnothing$, and that the hypotheses of either Theorem~\ref{thm:volcontraction_qubit} or Theorem~\ref{thm:volcontraction_nqubits} hold for $O_1$. Then the geometric information
\[
I_{\mathrm{HS}}(O_1,O_2)=\frac{\operatorname{Vol}_{\mathrm{HS}}(O_2)}{\operatorname{Vol}_{\mathrm{HS}}(O_1)}
\]
strictly increases under measurement:
\[
I_{\mathrm{HS}}(O_1'(\eta),O_2'(\eta))>I_{\mathrm{HS}}(O_1,O_2).
\]
\end{corollary}

The preceding volume-contraction and information-increase results concern
the ambient Hilbert--Schmidt volume and the corresponding information
$I_{\mathrm{HS}}$. No analogous monotonicity is asserted here for the
operational volume $\operatorname{Vol}_{\mathrm{op},V}$ or the
observational information $I_{\mathrm{op},V}$: these quantities are
defined relative to the finite-dimensional observational subspace of an
experimental parcel, whereas measurement preserves general quantum parcels
but need not preserve their experimental form. Consequently, after a
measurement update there need not be a canonically corresponding
observational subspace with respect to which the pre- and post-measurement
operational volumes can be compared.

The results of this section establish the core measurement-theoretic
foundation of IQM. Proposition~\ref{prop:singleupdate} guarantees that single-parcels are
preserved under fuzzy measurement, and Proposition~\ref{prop:monotone} shows the update formulas are monotone with respect to set inclusion, and hence order-compatible whenever the updated pair remains a double-parcel. Theorem~\ref{thm:double-preservation} gives a useful sufficient condition for a double-parcel to survive
measurement; Examples~\ref{ex:preserved-no-separation} and~\ref{ex:failure-no-separation} show that this condition is not necessary for preservation and that its failure alone does not determine whether preservation occurs. Theorems~\ref{thm:volcontraction_qubit} and~\ref{thm:volcontraction_nqubits} establish volume contraction of the
possible set, and Corollary~\ref{cor:strict-increase} shows that $I_{\mathrm{HS}}(O_1, O_2)$ increases
strictly whenever the double-parcel structure is preserved and the corresponding volume-contraction hypotheses hold. These results are applied
to the three foundational paradoxes in Sections~\ref{sec:double-slit},~\ref{sec:cat} and~\ref{sec:entanglement}. A further
application, combining the measurement framework with Reimann's
thermalisation theorem to derive IQM-native thermalisation results
for double parcels, is developed in a companion paper.

%% file: Uncertainty-and-Entropy.tex
\section{Uncertainty and Entropy}\label{sec:uncertainty}

In Interval Quantum Mechanics, when observables $A,B$ are measured on states represented by a quantum parcel $P\subset\mathcal D(\mathcal H)$, their standard deviations become range-valued.

\begin{definition}[Standard deviation range]
\label{def:stddev}
For a quantum parcel $P$ and a bounded self-adjoint operator $A$, define
\[
\Delta_P(A) := \{\Delta_\rho A : \rho\in P\},
\qquad
\Delta_\rho A := \sqrt{\operatorname{Tr}(\rho A^2)-[\operatorname{Tr}(\rho A)]^2}.
\]
Since $P$ is convex, hence connected, and $\rho\mapsto\Delta_\rho A$ is continuous, $\Delta_P(A)$ is an interval — open, closed, half-open, or a singleton, depending on $P$ and $A$ — by the same argument used for expectation and probability ranges (Sections~\ref{sec:reduction}--\ref{sec:unitary}). We record its closed enclosure
\[
\overline\Delta_P(A) := \Bigl[\inf_{\rho\in P}\Delta_\rho A,\ \sup_{\rho\in P}\Delta_\rho A\Bigr].
\]
\end{definition}

\begin{theorem}[Uncertainty relation for quantum parcels]
\label{thm:uncertainty}
Let $P\subset\mathcal D(\mathcal H)$ be a quantum parcel and $A,B$ bounded self-adjoint operators. Then
\[
\inf_{\rho\in P}\bigl(\Delta_\rho A\cdot\Delta_\rho B\bigr)
\ \ge\
\frac12\inf_{\rho\in P}\bigl|\operatorname{Tr}(\rho[A,B])\bigr|.
\]
This follows by taking the infimum over $\rho\in P$ in the Robertson inequality $\Delta_\rho A\cdot\Delta_\rho B\ge\frac12|\operatorname{Tr}(\rho[A,B])|$, which holds pointwise for every $\rho$.
\end{theorem}

Thus the uncertainty principle constrains the entire parcel: no state compatible with the finite observational specification can violate the Robertson bound. IQM therefore represents uncertainty not by assigning exact variances to an exact state, but by ranges of variances over the compatible state region.

A complementary measure of uncertainty is the von Neumann entropy
\[
S(\rho) = -\operatorname{Tr}(\rho\log\rho).
\]
For a quantum parcel $P$, define its entropy range by
\[
S_P := \{S(\rho):\rho\in P\},
\]
with closed enclosure
\[
\overline S_P := \Bigl[\inf_{\rho\in P}S(\rho),\ \sup_{\rho\in P}S(\rho)\Bigr].
\]
In finite dimensions $S$ is continuous, so the infimum and supremum over $P$ equal the minimum and maximum over the compact closure $\overline P$. Since von Neumann entropy is concave, its maximum over a convex parcel need not occur at an extreme point; consequently vertex enumeration of a polyhedral outer approximation does not in general determine the upper entropy bound.

The entropy range and geometric information quantify different aspects of
finite-precision information. The entropy range $S_P$ describes the
mixedness of the states compatible with a parcel, whereas
\[
I_{\mathrm{HS}}(O)
=
\frac{1}{\operatorname{Vol}_{\mathrm{HS}}(O)},
\qquad
I_{\mathrm{HS}}(O_1,O_2)
=
\frac{\operatorname{Vol}_{\mathrm{HS}}(O_2)}
     {\operatorname{Vol}_{\mathrm{HS}}(O_1)}
\]
measure, respectively, the inverse size of a possible region and the
relative sizes of the excluded and possible regions. No monotonicity of
$S_P$ under measurement is asserted. By contrast, under the
volume-contraction hypotheses of Theorem~\ref{thm:volcontraction_qubit}
or Theorem~\ref{thm:volcontraction_nqubits},
$I_{\mathrm{HS}}(O)$ increases strictly for a single parcel, while
Corollary~\ref{cor:strict-increase} gives strict increase of
$I_{\mathrm{HS}}(O_1,O_2)$ whenever the updated pair remains a
double-parcel.

This distinction is already visible for a pure superposition subjected to
a selective sharp measurement. The initial pure state has von Neumann
entropy zero, and the pure post-measurement outcome state also has entropy
zero, even though the measurement has supplied information about the
system. Geometric information is designed to capture this gain through
the contraction of the region of states compatible with the available
information, rather than through a change in their von Neumann entropy.

%% file: particle-wave.tex
\section{The Double-Slit Experiment in IQM}
\label{sec:double-slit}

Consider the two-dimensional path space spanned by $\ket{L}$ and $\ket{R}$.
Detection at a screen position with relative phase $\phi$ is represented by
\[
\Pi_\phi=\ket{\phi}\bra{\phi},
\qquad
\ket{\phi}
=
\frac{1}{\sqrt2}
\bigl(\ket L+e^{i\phi}\ket R\bigr).
\]
A finite-precision preparation is represented not by an exactly specified
density matrix but by an experimental parcel $O$.  For example, a preparation
centred on
\[
\ket{+}=\frac{\ket L+\ket R}{\sqrt2}
\]
may be represented by
\[
O_\varepsilon
=
\left\{
\rho\in\mathcal D(\mathbb C^2):
|\operatorname{Tr}(\rho\sigma_x)-1|<\varepsilon,\;
|\operatorname{Tr}(\rho\sigma_y)|<\varepsilon,\;
|\operatorname{Tr}(\rho\sigma_z)|<\varepsilon
\right\}.
\]
The corresponding screen probabilities are therefore the ranges
\[
\Pr_{O_\varepsilon}(\Pi_\phi)
=
\left\{
\operatorname{Tr}(\rho\Pi_\phi):
\rho\in O_\varepsilon
\right\}.
\]

A finite-resolution which-path measurement is described by the fuzzy POVM
of Section~\ref{sec:measurement},
\[
E_L^{(\eta)}
=
\eta\Pi_L+\frac{1-\eta}{2}I,
\qquad
E_R^{(\eta)}
=
\eta\Pi_R+\frac{1-\eta}{2}I,
\]
where $\Pi_L=\ket L\bra L$ and $\Pi_R=\ket R\bra R$.
Conditioned on outcome $L$, the preparation parcel becomes
\[
O_\varepsilon'
=
f_L^{(\eta)}(O_\varepsilon),
\]
and the subsequent screen prediction is
\[
\Pr_{O_\varepsilon'}(\Pi_\phi)
=
\left\{
\operatorname{Tr}(\rho\Pi_\phi):
\rho\in O_\varepsilon'
\right\}.
\]

The loss of interference can be seen directly from the qubit update.
Writing
\[
\rho=\frac12(I+x\sigma_x+y\sigma_y+z\sigma_z),
\]
Section~\ref{sec:measurement} gives, for outcome $L$,
\[
x'
=
\frac{\sqrt{1-\eta^2}\,x}{1+\eta z},
\qquad
y'
=
\frac{\sqrt{1-\eta^2}\,y}{1+\eta z},
\qquad
z'
=
\frac{z+\eta}{1+\eta z}.
\]
Since
\[
\operatorname{Tr}(\rho\Pi_\phi)
=
\frac12\bigl(1+x\cos\phi+y\sin\phi\bigr),
\]
the interference term is carried by the transverse components $x$ and $y$. On any parcel $O$ satisfying $\operatorname{Tr}(\rho\Pi_L)\ge c>0$, i.e.\ $z\ge2c-1$, the denominator obeys $1+\eta z\ge1+\eta(2c-1)$, so
\[
|x'|,|y'|
\ \le\
\frac{\sqrt{1-\eta^2}}{1+\eta(2c-1)}
\longrightarrow0
\qquad(\eta\to1^-),
\]
uniformly on $O$. Consequently
\[
\Pr_{f_L^{(\eta)}(O)}(\Pi_\phi)
\longrightarrow
\left\{\frac12\right\}.
\]
Thus increasingly sharp which-path information continuously suppresses the
interference pattern, with the usual sharp-measurement result recovered in
the limit.

The same experiment may be represented by a double parcel $(O_1,O_2)$.
After outcome $L$ its update is, by Definition~\ref{def:update},
\[
O_1'=f_L^{(\eta)}(O_1),
\qquad
O_2'
=
\operatorname{Conv}
\bigl(O_2\cup f_L^{(\eta)}(O_2)\bigr),
\]
whenever the updated sets remain disjoint. The double-parcel
description therefore records not only the range of states compatible with
the preparation and measurement history but also the states already excluded
by that history.

In this formulation the familiar wave--particle duality paradox does not
arise in its standard form. A single finite-precision parcel carries the
range of coherent path information compatible with the preparation, while
the measurement update continuously suppresses the transverse coherences
as which-path resolution increases. There is therefore no need to regard
the system as switching between two different kinds of quantum state,
``particle-like'' and ``wave-like'': the apparent duality is represented
geometrically by the continuous evolution of a single parcel as the
measurement resolution changes. Standard sharp-measurement complementarity
is recovered in the limit $\eta\to1^-$.

%% file: cat-paradox.tex
\section{Schr\"odinger's Cat in IQM}
\label{sec:cat}

Consider first the two-dimensional subspace spanned by macroscopically
distinguishable states $\ket{\text{alive}}$ and $\ket{\text{dead}}$.  A coherent
superposition
\[
\ket{\psi}
=
\alpha\ket{\text{alive}}+\beta\ket{\text{dead}},
\qquad
|\alpha|^2+|\beta|^2=1,
\]
corresponds in standard quantum mechanics to the point state
$\ket{\psi}\bra{\psi}$.  In IQM a finite-precision preparation centred on
this state is represented instead by an experimental parcel containing it.
For a qubit realization, with $\ket{\text{alive}}=\ket0$ and
$\ket{\text{dead}}=\ket1$, one may take
\[
O_\varepsilon
=
\left\{
\rho\in\mathcal D(\mathbb C^2):
|\operatorname{Tr}(\rho\sigma_x)-x_0|<\varepsilon,\;
|\operatorname{Tr}(\rho\sigma_y)-y_0|<\varepsilon,\;
|\operatorname{Tr}(\rho\sigma_z)-z_0|<\varepsilon
\right\},
\]
where $(x_0,y_0,z_0)$ is the Bloch vector of
$\ket{\psi}\bra{\psi}$.  Thus finite observational information need not
distinguish an exactly coherent state from nearby partially decohered
states.

A finite-resolution measurement of the cat's status is represented by the
fuzzy POVM of Section~\ref{sec:measurement},
\[
E_{\text{alive}}^{(\eta)}
=
\eta\Pi_{\text{alive}}+\frac{1-\eta}{2}I,
\qquad
E_{\text{dead}}^{(\eta)}
=
\eta\Pi_{\text{dead}}+\frac{1-\eta}{2}I.
\]
Conditioned on outcome ``alive'', the possible parcel becomes
\[
O'=f_{\text{alive}}^{(\eta)}(O).
\]

For the qubit model the behaviour follows directly from the update formula
already obtained in Section~\ref{sec:measurement}.  If
\[
\rho=\frac12(I+x\sigma_x+y\sigma_y+z\sigma_z),
\]
then the alive update gives
\[
x'
=
\frac{\sqrt{1-\eta^2}\,x}{1+\eta z},
\qquad
y'
=
\frac{\sqrt{1-\eta^2}\,y}{1+\eta z},
\qquad
z'
=
\frac{z+\eta}{1+\eta z}.
\]
Hence, on any parcel satisfying
$\operatorname{Tr}(\rho\Pi_{\text{alive}})\ge c>0$,
\[
|x'|,|y'|
\le
\frac{\sqrt{1-\eta^2}}
     {1+\eta(2c-1)}
\longrightarrow0,
\]
and
\[
1-z'
=
\frac{(1-\eta)(1-z)}
     {1+\eta z}
\longrightarrow0
\qquad(\eta\to1^-)
\]
uniformly on the parcel.  Thus
\[
\sup_{\rho\in O}
\left\|
f_{\text{alive}}^{(\eta)}(\rho)
-
\ket{\text{alive}}\bra{\text{alive}}
\right\|_1
\longrightarrow0
\qquad(\eta\to1^-).
\]
The analogous statement holds for outcome ``dead''.  The sharp-measurement behaviour of the usual cat
experiment is therefore recovered continuously from finite-resolution
parcel updates.

The same construction does not depend on the cat being literally a qubit.
For a macroscopic system, $\Pi_{\text{alive}}$ may project onto a large
macroscopic ``alive'' subspace rather than a single vector. A finite
collection of coarse-grained observables defines an experimental parcel in
the full state space, and conditioning on the observed outcome updates it
by the same Kraus rule. As the measurement becomes sharp, states with
nonzero probability for that outcome are driven towards states supported
on the corresponding macroscopic subspace. For rank-one $\Pi_{\text{alive}}$
this gives the unique state $\ket{\text{alive}}\bra{\text{alive}}$, whereas
for higher-rank $\Pi_{\text{alive}}$ the sharp map
\[
\rho\mapsto
\frac{\Pi_{\text{alive}}\rho\Pi_{\text{alive}}}
{\operatorname{Tr}(\rho\Pi_{\text{alive}})}
\]
generally retains dependence on $\rho$ within the alive subspace.

For a double parcel $(O_1,O_2)$ the same measurement is represented by
\[
O_1'=f_{\text{alive}}^{(\eta)}(O_1),
\qquad
O_2'
=
\operatorname{Conv}
\bigl(O_2\cup f_{\text{alive}}^{(\eta)}(O_2)\bigr),
\]
whenever the updated sets remain disjoint. The second component thereby
records states excluded by the accumulated preparation and measurement
history.

IQM therefore reformulates the cat experiment so that the familiar
Schr\"odinger-cat paradox does not arise in its standard form. A
finite-precision preparation does not warrant assigning the system an
exact coherent superposition: coherent and partially decohered states may
belong to the same experimentally specified parcel. When a macroscopic
outcome is observed, the measurement updates this parcel continuously
toward the corresponding outcome sector as the measurement resolution
increases. Thus the apparent paradox of an exactly specified macroscopic
superposition undergoing an abrupt collapse is replaced by the geometric
evolution and refinement of a finite-precision parcel, with the standard
sharp-measurement behaviour recovered in the limiting case.

%% file: entanglement.tex
\section{Entanglement and Bell Violation in IQM}
\label{sec:entanglement}

Consider two qubits in the Bell state
\[
\ket{\Phi^+}
=
\frac{1}{\sqrt2}(\ket{00}+\ket{11}).
\]
In IQM, a finite-precision preparation centred on this state is
represented by an experimental parcel containing it. We show that such a
parcel can be chosen so that every compatible state violates CHSH.

For the standard CHSH observables
\[
A=\sigma_x,\qquad
A'=\sigma_z,\qquad
B=\frac{\sigma_x+\sigma_z}{\sqrt2},
\qquad
B'=\frac{\sigma_x-\sigma_z}{\sqrt2},
\]
let
\[
S
=
A\otimes B+A\otimes B'
+A'\otimes B-A'\otimes B'.
\]
The Bell state satisfies
\[
\bra{\Phi^+}S\ket{\Phi^+}=2\sqrt2.
\]

\begin{theorem}[Finite-precision Bell violation]
\label{bell-violation}
There exists an experimental parcel $P$ containing
$\ket{\Phi^+}\bra{\Phi^+}$ such that
\[
E_P(S)\subset(2,2\sqrt2].
\]
Consequently every state in $P$ violates the CHSH inequality and is
entangled.
\end{theorem}

More generally, finite precision does not prevent entanglement from being
certified by a witness.

\begin{corollary}[Entanglement witness]
Let $W$ be an entanglement witness and suppose
\[
\operatorname{Tr}(\rho_0W)<0.
\]
Then there exists an experimental parcel $P\ni\rho_0$ such that
\[
\operatorname{Tr}(\rho W)<0
\qquad\text{for every }\rho\in P.
\]
This follows by the same openness argument.
\end{corollary}

A local finite-resolution measurement on Alice's subsystem is represented
by the fuzzy POVM and parcel update of Section~\ref{sec:measurement}.

Thus CHSH violation is not an artefact of an infinitely precise Bell
point state: it persists throughout an open finite-precision parcel.

%% file: computational-aspects.tex
\section{Hyper-rectangular observational representations}
\label{sec:computational}

The state-space outer approximation introduced in
Section~\ref{sec:parcels} has a particularly simple representation in
finite observational coordinates. We describe this representation here
and show how it can be used to obtain computationally inexpensive bounds
on expectation values.

Let $H_1,\ldots,H_m$ be a Hilbert--Schmidt-orthonormal basis of an
$m$-dimensional observational subspace
\[
V=\operatorname{span}_{\mathbb R}\{H_1,\ldots,H_m\},
\]
chosen as in Section~\ref{sec:parcels}; in finite dimensions the $H_j$
are traceless Hermitian operators. The associated observational
coordinate map is
\[
\Phi_V(\rho)
=
\bigl(
\operatorname{Tr}(\rho H_1),\ldots,
\operatorname{Tr}(\rho H_m)
\bigr)
\in\mathbb R^m.
\]
For intervals $(a_j,b_j)$, let
\[
R(a,b)
=
\prod_{j=1}^m(a_j,b_j)
\subset\mathbb R^m.
\]
The corresponding experimental parcel is
\[
O(a,b)
=
\{\rho\in\mathcal D(\mathcal H):
\Phi_V(\rho)\in R(a,b)\},
\]
and its observational shadow is
\[
\Phi_V(O(a,b))
=
\Phi_V(\mathcal D(\mathcal H))\cap R(a,b).
\]
Thus positivity generally makes the physical observational shadow a
proper subset of the coordinate box $R(a,b)$.

The state-space outer approximation defined in
Section~\ref{sec:parcels} is obtained by dropping positivity while
retaining the trace-one, self-adjoint and observational constraints.
In the coordinates above this amounts to allowing the whole box
$R(a,b)$ rather than only those points that arise from positive density
matrices. Thus
\[
\Phi_V(O(a,b))
=
\Phi_V(\mathcal D(\mathcal H))\cap R(a,b)
\subseteq R(a,b),
\]
whereas the state-space outer approximation has the full box $R(a,b)$
as its observational coordinate image. We therefore call $R(a,b)$ the
\emph{hyper-rectangular observational representation} of the
state-space outer approximation. It requires only the $2m$ endpoints
$(a_j,b_j)$ to store, while retaining all the defining expectation
intervals.

\[
\begin{tikzcd}[column sep=large,row sep=large]
O(a,b)
    \arrow[r, hook]
    \arrow[d, "\Phi_V"']
&
O^{\mathrm{outer}}(a,b)
    \arrow[d, two heads, "\Phi_V"]
\\
\Phi_V(\mathcal D(\mathcal H))\cap R(a,b)
    \arrow[r, hook]
&
R(a,b).
\end{tikzcd}
\]

This coordinate-box construction is meaningful whether $\mathcal H$ is
finite- or infinite-dimensional: only finitely many observational
directions are involved, so $\Phi_V$ always takes values in the
finite-dimensional space $\mathbb R^m$. In infinite dimensions the
corresponding outer approximation in the full trace-one self-adjoint
space is unbounded, as discussed in Section~\ref{sec:parcels}, but its
finite observational representation remains the bounded box $R(a,b)$.
We now specialize to $\dim\mathcal H=d<\infty$ in order to relate this
representation to the full Hilbert--Schmidt geometry and to obtain
explicit expectation-value bounds.

For an observable $A\in\mathbb RI\oplus V$, write
\[
A
=
\alpha I+\sum_{j=1}^m c_jH_j,
\qquad
\alpha=\frac{\operatorname{Tr}A}{d},
\qquad
c_j=\operatorname{Tr}(AH_j).
\]
This decomposition is unique because $I$ is Hilbert--Schmidt orthogonal
to the traceless operators $H_j$. For every state $\rho$,
\[
\operatorname{Tr}(\rho A)
=
\alpha+\sum_{j=1}^m
c_j\operatorname{Tr}(\rho H_j).
\]
Consequently the coordinate box gives the enclosure
\[
\overline E_{O(a,b)}(A)
\subseteq
\overline E^{\,\mathrm{box}}_{a,b}(A)
:=
\left[
\alpha+\sum_j(c_j^+a_j+c_j^-b_j),
\;
\alpha+\sum_j(c_j^+b_j+c_j^-a_j)
\right],
\]
where
\[
c_j^+=\max(c_j,0),
\qquad
c_j^-=\min(c_j,0).
\]
This bound is obtained directly from the coordinate intervals and
requires no vertex enumeration. Since the physical shadow
$\Phi_V(\mathcal D(\mathcal H))\cap R(a,b)$ may be strictly smaller
than $R(a,b)$, the box enclosure can be conservative relative to the
true expectation range $\overline E_{O(a,b)}(A)$.

Under unitary evolution, let
\[
\sigma=U\rho U^\dagger.
\]
Then
\[
\operatorname{Tr}(\rho H_j)
=
\operatorname{Tr}(\sigma\,UH_jU^\dagger).
\]
Hence the evolved parcel is represented by the same coordinate box
attached to the transported observational subspace
\[
V'=UVU^\dagger.
\]
Equivalently, if the output observables $H_j$ are held fixed, their
expectation values on the evolved state may be evaluated on the incoming
state using the Heisenberg-transformed observables $U^\dagger H_jU$:
\[
\operatorname{Tr}(U\rho U^\dagger H_j)
=
\operatorname{Tr}(\rho U^\dagger H_jU).
\]
These are the Schr\"odinger and Heisenberg descriptions, respectively,
of the same transformation. The former is the coordinate counterpart
of the transported operational volume discussed in
Section~\ref{sec:unitary}.

Measurement updates are nonlinear and need not preserve experimental
parcels or their hyper-rectangular observational representations.
Consequently there is, in general, no canonical post-measurement box
obtained simply by transporting the original one. A practical outer
representation may instead be constructed by choosing a finite family
of post-measurement observables and recomputing bounds on their
expectation values, for example by optimization over the original
parcel. The resulting intervals again determine a hyper-rectangular
observational representation. Developing efficient algorithms for this
propagation is beyond the mathematical parcel formalism considered here.

Hyper-rectangular representations therefore provide a simple finite data
structure for the observational content of an experimental parcel. The
coordinate box $R(a,b)$ is the hyper-rectangular observational
representation of its state-space outer approximation, whereas
\[
\Phi_V(\mathcal D(\mathcal H))\cap R(a,b)
\]
is the actual observational shadow after positivity of the quantum state
has been imposed.

%% file: conclusion.tex
\section{Conclusion and outlook}
\label{sec:conclusion}

We have developed Interval Quantum Mechanics (IQM) as a finite-precision
formulation of quantum theory in which an experimentally specified state is
represented by a quantum parcel rather than by an exactly specified density
operator. Experimental parcels are convex open sets determined by finitely
many expectation constraints; point states remain part of the mathematical
theory and are recovered as ideal limits of increasingly refined
observational information.

The main results establish that this replacement is compatible with the
standard structure of quantum mechanics while introducing a genuinely
finite-resolution level of description. Experimental parcels form a basis
for the state-space topology, and nested refinements recover point states in
both finite and infinite dimensions. Bounded-observable predictions converge
to their standard point-state values, while the corresponding result for
non-negative unbounded observables holds under a uniform moment condition.
Unitary evolution acts reversibly and preserves ambient Hilbert--Schmidt
volume in finite dimensions. Finite-resolution measurements are represented
by fuzzy POVMs and Kraus updates: single parcels are preserved, the revised
double-parcel update is order-compatible whenever it remains within the
double-parcel domain, and a useful sufficient condition guarantees preservation for sufficiently sharp measurements.
For sufficiently sharp measurements we established contraction of ambient
Hilbert--Schmidt volume under the stated hypotheses and the corresponding
strict increase of geometric information.

The framework also gives finite-precision formulations of familiar
foundational phenomena. Which-path measurements continuously suppress the
transverse coherences responsible for double-slit interference. In the
Schr\"odinger-cat model, finite-resolution preparation and measurement are
represented without assigning an experimentally inaccessible exact state,
with sharp measurements concentrating the parcel on the observed outcome
sector. Most importantly, Bell violation survives finite precision: an
entire open experimental parcel around a Bell state can violate CHSH.
Thus quantum nonclassicality is not an artefact of an infinitely precise
point-state description.

The distinction between parcels and point states should not be understood as
a rejection of standard quantum mechanics. Rather, IQM separates two levels
of description: exact density operators provide the limiting mathematical
theory, while parcels represent finite observational specifications. The
reduction results make this relation precise — an instance, as discussed in
the Introduction, of the general relation between finite numerical
representations and exact real-number analysis.

The same viewpoint has begun to extend beyond the setting developed here.
An algebraic formulation of finite-precision quantum field theory has been
developed in~\cite{EdalatIAQFT2026}, and quantum ergodic and thermalisation
results in the parcel framework are studied in~\cite{EdalatErgodicity2026}.
Together these
developments suggest that the parcel construction is not tied to the
finite-dimensional foundational examples considered in this paper, but can
serve as a common finite-precision language across different parts of quantum
theory.

Several questions remain for further work. The sufficient ratio-form condition for
preservation of double parcels under measurement is not necessary, and a
more intrinsic characterization of useful preservation conditions would be
valuable. Efficient propagation of finite-dimensional observational
representations under general measurements also remains to be developed.
More broadly, the relation between finite observational parcels and the
operational precision attainable in concrete experiments deserves systematic
study.

A related question is empirical rather than purely mathematical. IQM and
standard point-state quantum mechanics agree in the infinite-precision
limit, and at any finite stage IQM's ranges are constructed to contain
whatever standard quantum mechanics would predict for a compatible point
state; identifying a finite-precision experimental setting in which the two
descriptions could in principle be told apart, rather than merely restated
in different language, is left for future work.

IQM therefore provides a mathematically explicit separation between exact
quantum states and finite observational information. Its principal claim is
that finite precision can be incorporated into the
state-space formalism itself, while standard point-state quantum mechanics
is retained as the limiting theory.

%% file: appendix.tex
\section*{Appendix: Proofs}

\subsection*{Proofs from Section~\ref{sec:parcels} (Quantum parcels)}

\noindent{\bf Proposition}~{\bf \ref{prop:no-open-pure-states}}~(No open set of exclusively pure states)

\begin{proof}
Suppose, for contradiction, that $O$ is a non-empty weak open subset of $\mathcal D(\mathcal H)$ consisting exclusively of pure states, and fix $\rho_0\in O$. Since $\dim\mathcal H\ge2$, $\mathcal D(\mathcal H)$ contains a state $\sigma\ne\rho_0$. For $\varepsilon\in(0,1]$ define
\[
\rho_\varepsilon := (1-\varepsilon)\rho_0+\varepsilon\sigma.
\]
Since $\rho_0\ne\sigma$, each $\rho_\varepsilon$ with $\varepsilon\in(0,1)$ is a proper convex combination of two distinct elements of $\mathcal D(\mathcal H)$, hence is not an extreme point of $\mathcal D(\mathcal H)$; since the extreme points of $\mathcal D(\mathcal H)$ are exactly the pure states, $\rho_\varepsilon$ is not pure for any $\varepsilon\in(0,1)$.

On the other hand,
\[
\|\rho_\varepsilon-\rho_0\|_1=\varepsilon\|\sigma-\rho_0\|_1\longrightarrow0
\qquad(\varepsilon\to0^+),
\]
so $\rho_\varepsilon\to\rho_0$ in trace norm, hence — by the weak/trace-norm coincidence established at the start of this section — in the weak topology as well. Since $O$ is open and $\rho_0\in O$, there exists $\varepsilon_0\in(0,1)$ such that $\rho_\varepsilon\in O$ for all $\varepsilon\in(0,\varepsilon_0)$. Any such $\rho_\varepsilon$ is a non-pure element of $O$, contradicting the assumption that $O$ consists exclusively of pure states.
\end{proof}

\subsection*{Proofs from Section~\ref{sec:reduction} (Reduction to Standard Quantum Mechanics)}

\noindent{\bf Proposition}~{\bf \ref{prop:expectation}}~(Expectation ranges are intervals)

\begin{proof}
Let $\alpha=\inf_{\rho\in O}\operatorname{Tr}(\rho H)$ and $\beta=\sup_{\rho\in O}\operatorname{Tr}(\rho H)$. If $\alpha=\beta$, every $\rho\in O$ satisfies $\operatorname{Tr}(\rho H)=\alpha$, so $E_O(H)=\{\alpha\}$.

Assume $\alpha<\beta$ and let $r\in(\alpha,\beta)$. Since $r>\alpha=\inf_{\rho\in O}\operatorname{Tr}(\rho H)$, there exists $\rho_-\in O$ with $\operatorname{Tr}(\rho_-H)<r$; since $r<\beta=\sup_{\rho\in O}\operatorname{Tr}(\rho H)$, there exists $\rho_+\in O$ with $\operatorname{Tr}(\rho_+H)>r$. Since $O$ is convex, $\rho_t=t\rho_-+(1-t)\rho_+\in O$ for all $t\in[0,1]$, and $t\mapsto\operatorname{Tr}(\rho_tH)$ is continuous and affine, so by the intermediate value theorem there is $t_0\in(0,1)$ with $\operatorname{Tr}(\rho_{t_0}H)=r$. Hence every $r\in(\alpha,\beta)$ lies in $E_O(H)$, and together with $E_O(H)\subseteq[\alpha,\beta]$ (by definition of $\alpha,\beta$) this gives $(\alpha,\beta)\subseteq E_O(H)\subseteq[\alpha,\beta]$, i.e. $E_O(H)$ is an interval with endpoints $\alpha,\beta$. Whether $\alpha$ or $\beta$ themselves lie in $E_O(H)$ depends on whether some $\rho\in O$ attains them; both occur in examples (e.g.\ for $O=\mathcal D(\mathcal H)$ itself, or for $O=\{\rho:\operatorname{Tr}(\rho H)>c\}$ containing a pure state that maximizes $\operatorname{Tr}(\rho H)$ over $\mathcal D(\mathcal H)$).
\end{proof}

\noindent{\bf Theorem}~{\bf \ref{thm:reduction-finite}}~(Reduction to point states: finite dimensions)

\begin{proof}
For $\rho\in O_n(\rho_0)$, the difference $\rho-\rho_0$ is traceless and Hermitian, so by Parseval's identity in the Hilbert--Schmidt-orthonormal basis $\{\Psi_j\}_{j=1}^{N^2-1}$ of the traceless Hermitian operators (Section~\ref{sec:parcels}),
\[
\|\rho-\rho_0\|_{\mathrm{HS}}^2
=
\sum_{j=1}^{N^2-1}|\operatorname{Tr}((\rho-\rho_0)\Psi_j)|^2
<(N^2-1)\,4^{-n}.
\]
Since $\rho-\rho_0$ is an $N\times N$ Hermitian matrix, Cauchy--Schwarz applied to its (real) eigenvalues gives $\|\rho-\rho_0\|_1\le\sqrt N\,\|\rho-\rho_0\|_{\mathrm{HS}}$, so
\[
\|\rho-\rho_0\|_1 < \sqrt{N(N^2-1)}\,2^{-n} =: \varepsilon_n \xrightarrow[n\to\infty]{}0.
\]
For $\rho,\sigma\in O_n(\rho_0)$, the triangle inequality gives $\|\rho-\sigma\|_1<2\varepsilon_n$, so $\operatorname{diam}_1(O_n(\rho_0))\le2\varepsilon_n\to0$. Since $\rho_0\in O_n(\rho_0)$ for every $n$ and the diameters shrink to $0$, any element of $\bigcap_n O_n(\rho_0)$ lies within trace distance $\varepsilon_n$ of $\rho_0$ for every $n$, so $\bigcap_n O_n(\rho_0)=\{\rho_0\}$.

For any bounded Hermitian $H$ and $\rho\in O_n(\rho_0)$,
\[
|\operatorname{Tr}(\rho H)-\operatorname{Tr}(\rho_0H)| \le \|H\|\,\|\rho-\rho_0\|_1 < \|H\|\,\varepsilon_n,
\]
so $\sup_{\rho\in O_n(\rho_0)}|\operatorname{Tr}(\rho H)-\operatorname{Tr}(\rho_0H)|\to0$, giving the stated limits.
\end{proof}

\noindent{\bf Theorem}~{\bf \ref{thm:reduction-infinite-bounded}}~(Reduction to point states: infinite dimensions)

\begin{proof}
By Section~\ref{sec:parcels}, the weak topology on $\mathcal D(\mathcal H)$ coincides with the trace-norm topology, so for each $n$ the ball $B_1(\rho_0,2^{-n})$ is a weak open neighbourhood of $\rho_0$; since experimental parcels form a basis for this topology, every weak open neighbourhood of $\rho_0$ contains an experimental parcel containing $\rho_0$.

We construct the sequence $O_n(\rho_0)$ inductively. Choose an experimental parcel $O_1(\rho_0)$ with $\rho_0\in O_1(\rho_0)\subseteq B_1(\rho_0,2^{-1})$. Given $O_n(\rho_0)$, the set $O_n(\rho_0)\cap B_1(\rho_0,2^{-(n+1)})$ is a weak open neighbourhood of $\rho_0$, so there is an experimental parcel $O_{n+1}(\rho_0)$ with
\[
\rho_0\in O_{n+1}(\rho_0)\subseteq O_n(\rho_0)\cap B_1(\rho_0,2^{-(n+1)}).
\]
By construction $O_{n+1}(\rho_0)\subseteq O_n(\rho_0)$ and $O_{n+1}(\rho_0)\subseteq B_1(\rho_0,2^{-(n+1)})$, so the sequence is nested and $O_n(\rho_0)\subseteq B_1(\rho_0,2^{-n})$ for every $n$, giving $\sup_{\rho\in O_n(\rho_0)}\|\rho-\rho_0\|_1\le2^{-n}$. For bounded Hermitian $H$ and $\rho\in O_n(\rho_0)$,
\[
|\operatorname{Tr}(\rho H)-\operatorname{Tr}(\rho_0H)| \le \|H\|\,\|\rho-\rho_0\|_1 \le \|H\|\,2^{-n}
\xrightarrow[n\to\infty]{}0.
\]
\end{proof}

\noindent{\bf Theorem}~{\bf \ref{thm:unbounded-moment}}~(Reduction for unbounded observables under a uniform moment bound)

\begin{proof}
Let $\{E_\lambda\}_{\lambda\ge0}$ denote the spectral family of $H$, and define the bounded truncations
\[
H_R := H\,E_{[0,R]}(H)
\]
for $R>0$. Then $H_R$ is bounded and satisfies $0\le H_R\le H$.

We decompose
\[
\operatorname{Tr}(\rho H)-\operatorname{Tr}(\rho_0 H)
=
\bigl[\operatorname{Tr}(\rho H_R)-\operatorname{Tr}(\rho_0 H_R)\bigr]
+
\operatorname{Tr}\bigl(\rho(H-H_R)\bigr)
-
\operatorname{Tr}\bigl(\rho_0(H-H_R)\bigr).
\]
Hence
\begin{align}
\left|
\operatorname{Tr}(\rho H)-\operatorname{Tr}(\rho_0 H)
\right|
\le\;&
\left|
\operatorname{Tr}(\rho H_R)-\operatorname{Tr}(\rho_0 H_R)
\right|
\nonumber\\
&+
\operatorname{Tr}\bigl(\rho(H-H_R)\bigr)
+
\operatorname{Tr}\bigl(\rho_0(H-H_R)\bigr).
\label{eq:three-term-bound}
\end{align}
(The tail terms are non-negative since $H - H_R \ge 0$, so the absolute values may be dropped.)

\medskip
\noindent
\textbf{Step 1: Convergence for the bounded truncations.}
Since $H_R$ is bounded, Theorem~\ref{thm:reduction-infinite-bounded} implies that
\[
\lim_{n\to\infty}
\sup_{\rho\in O_n(\rho_0)}
\left|
\operatorname{Tr}(\rho H_R)-\operatorname{Tr}(\rho_0 H_R)
\right|
=0
\]
for every fixed $R>0$.

\medskip
\noindent
\textbf{Step 2: Uniform control of the tails.}
Observe that $H-H_R = H\,E_{(R,\infty)}(H)$. Since $\lambda \le \lambda^{1+\delta}/R^\delta$ for all $\lambda>R$, the functional calculus gives $H\,E_{(R,\infty)}(H)\le R^{-\delta} H^{1+\delta}$. Therefore, for every $\rho\in O_n(\rho_0)$,
\[
\operatorname{Tr}\bigl(\rho(H-H_R)\bigr)
\le
R^{-\delta}\operatorname{Tr}(\rho H^{1+\delta}).
\]
By the uniform moment assumption, there exists a constant $C<\infty$ such that
\[
\operatorname{Tr}\bigl(\rho(H-H_R)\bigr)
\le
C R^{-\delta}
\]
uniformly in $n$ and $\rho\in O_n(\rho_0)$. Hence
\[
\lim_{R\to\infty}
\sup_n\sup_{\rho\in O_n(\rho_0)}
\operatorname{Tr}\bigl(\rho(H-H_R)\bigr)
=0.
\]
Similarly, because $\operatorname{Tr}(\rho_0 H)<\infty$, monotone convergence implies $\lim_{R\to\infty}\operatorname{Tr}(\rho_0(H-H_R))=0$.

\medskip
\noindent
\textbf{Step 3: Conclusion.}
Let $\varepsilon>0$. Choose $R>0$ large enough that
\[
\sup_n\sup_{\rho\in O_n(\rho_0)}\operatorname{Tr}(\rho(H-H_R))<\frac\varepsilon3,
\qquad
\operatorname{Tr}(\rho_0(H-H_R))<\frac\varepsilon3.
\]
For this fixed $R$, choose $n_0$ such that for all $n\ge n_0$,
\[
\sup_{\rho\in O_n(\rho_0)}|\operatorname{Tr}(\rho H_R)-\operatorname{Tr}(\rho_0 H_R)|<\frac\varepsilon3.
\]
Then Equation~\eqref{eq:three-term-bound} gives, for all $n\ge n_0$ and all $\rho\in O_n(\rho_0)$,
\[
|\operatorname{Tr}(\rho H)-\operatorname{Tr}(\rho_0 H)|<\varepsilon.
\]
Therefore $\lim_{n\to\infty}\sup_{\rho\in O_n(\rho_0)}|\operatorname{Tr}(\rho H)-\operatorname{Tr}(\rho_0 H)|=0$, and the convergence of the upper and lower expectation bounds follows immediately.
\end{proof}

\subsection*{Proofs from Section~\ref{sec:unitary} (Interval Unitary Dynamics)}

\noindent{\bf Theorem}~{\bf \ref{thm:unitary-evolution}}~(Properties of unitary evolution)

\begin{proof}
(i) Reversibility follows immediately from $U^{-1} = U^\dagger$ and the
definition of the evolved parcel.

(ii) Preservation of the information order follows from the fact that
$\rho \mapsto U\rho U^\dagger$ is a bijection preserving set inclusion: if
$O'\subseteq O$ then $U(O')\subseteq U(O)$, i.e.\ $O\sqsubseteq O'$
implies $U(O)\sqsubseteq U(O')$; similarly for double parcels.

(iii) For traceless Hermitian $X,Y$,
\[
\langle UXU^\dagger,\,UYU^\dagger\rangle_{\mathrm{HS}}
=
\operatorname{Tr}\bigl((UXU^\dagger)^\dagger UYU^\dagger\bigr)
=
\operatorname{Tr}(UXU^\dagger UYU^\dagger)
=
\operatorname{Tr}(XY)
=
\langle X,Y\rangle_{\mathrm{HS}},
\]
using $X^\dagger=X$, cyclicity of the trace, and $U^\dagger U=I$. Thus $X\mapsto UXU^\dagger$ is an orthogonal transformation of the real Hilbert space of traceless Hermitian operators. Writing $\rho=I/N+X$ for the trace-one affine hyperplane (finite $N=\dim\mathcal H$), $U\rho U^\dagger=I/N+UXU^\dagger$ since $U(I/N)U^\dagger=I/N$: the map $\rho\mapsto U\rho U^\dagger$ fixes the centre $I/N$ of the hyperplane and acts on it via the orthogonal transformation $X\mapsto UXU^\dagger$ alone — no additional translation is involved. An orthogonal linear transformation has Jacobian of absolute value $1$, hence
\[
\operatorname{Vol}_{\mathrm{HS}}(U(O))=\operatorname{Vol}_{\mathrm{HS}}(O).
\]
\end{proof}

\subsection*{Proofs from Section~\ref{sec:measurement} (Measurement in Interval Quantum Mechanics)}

\noindent{\bf Proposition}~{\bf \ref{prop:singleupdate}}~(Measurement update for a single-parcel)

\begin{proof}
 We need to show that $O'$ is weak open.
The map $f_j^{(\eta)}: \mathcal{D}(\mathcal{H}) \to \mathcal{D}(\mathcal{H})$, is a composition of the linear map $\rho \mapsto M_j^{(\eta)} \rho (M_j^{(\eta)})^\dagger$ and the normalization by the positive scalar function $\rho \mapsto \operatorname{Tr}(\rho E_j^{(\eta)})$. Since $0<\eta<1$, the term $(1-\eta)I/m$ makes
$E_j^{(\eta)}$ positive definite, so the denominator is strictly positive. Moreover, $M_j^{(\eta)}$ is invertible, so $f_j^{(\eta)}$ is a homeomorphism: its inverse is given by \[\tau \mapsto
\frac{
(M_j^{(\eta)})^{-1}\tau((M_j^{(\eta)})^\dagger)^{-1}
}{
\operatorname{Tr}\!\left(
(M_j^{(\eta)})^{-1}\tau((M_j^{(\eta)})^\dagger)^{-1}
\right)
}.\] Therefore, $f_j^{(\eta)}$ maps open sets to open sets, and $O'$ is open. It remains to show that $f_j^{(\eta)}$ preserves convexity.
Let $\rho_1,\rho_2\in O$ and $t\in[0,1]$, set
$\rho_t=t\rho_1+(1-t)\rho_2$ and $p_k=\operatorname{Tr}(\rho_k E_j^{(\eta)})$
for $k=1,2$. A direct computation gives
\[
f_j^{(\eta)}(\rho_t)
=\frac{tp_1}{tp_1+(1-t)p_2}\,f_j^{(\eta)}(\rho_1)
+\frac{(1-t)p_2}{tp_1+(1-t)p_2}\,f_j^{(\eta)}(\rho_2).
\]
The coefficients are non-negative and sum to $1$, so $f_j^{(\eta)}(\rho_t)$
is a convex combination of $f_j^{(\eta)}(\rho_1)$ and $f_j^{(\eta)}(\rho_2)$.
Hence $f_j^{(\eta)}$ maps convex sets to convex sets, and $O'$ is convex.
Thus $O'$ is a valid quantum parcel.
\end{proof}

\noindent{\bf Proposition}~{\bf \ref{prop:monotone}}~(Monotonicity of the update formulas)

\begin{proof}
Assume $U_1\subseteq O_1$ and $O_2\subseteq U_2$ (this holds in particular whenever $(O_1,O_2)\sqsubseteq(U_1,U_2)$).

Since $f_j^{(\eta)}$ is a function, it preserves inclusion: $U_1'=f_j^{(\eta)}(U_1)\subseteq f_j^{(\eta)}(O_1)=O_1'$.

Also $O_2\subseteq U_2$ gives $f_j^{(\eta)}(O_2)\subseteq f_j^{(\eta)}(U_2)$, hence
\[
O_2\cup f_j^{(\eta)}(O_2)\ \subseteq\ U_2\cup f_j^{(\eta)}(U_2),
\]
and since the convex hull is monotone with respect to set inclusion,
\[
O_2'=\operatorname{Conv}\bigl(O_2\cup f_j^{(\eta)}(O_2)\bigr)\ \subseteq\ \operatorname{Conv}\bigl(U_2\cup f_j^{(\eta)}(U_2)\bigr)=U_2'.
\]
Thus $U_1'\subseteq O_1'$ and $O_2'\subseteq U_2'$; whenever both $(O_1',O_2')$ and $(U_1',U_2')$ are themselves double-parcels, these are precisely the two inclusions defining $(O_1',O_2')\sqsubseteq(U_1',U_2')$.
\end{proof}

\noindent{\bf Theorem}~{\bf \ref{thm:double-preservation}}~(A sufficient condition for preservation of double-parcels)

\begin{proof}
\textbf{Positivity is automatic.} Suppose $\rho\in\overline{O_1}\cup\overline{O_2}$ has $\operatorname{Tr}(\rho\Pi_j)=0$. Since $\rho\ge0$, $\Pi_j\rho\Pi_j$ is positive semidefinite with $\operatorname{Tr}(\Pi_j\rho\Pi_j)=\operatorname{Tr}(\rho\Pi_j)=0$ (cyclicity of the trace), so $\Pi_j\rho\Pi_j=0$ (a positive semidefinite operator of trace zero is the zero operator). Then, for $H=\Pi_jH\Pi_j$, $\operatorname{Tr}(\rho H)=\operatorname{Tr}(\rho\Pi_jH\Pi_j)=\operatorname{Tr}(\Pi_j\rho\Pi_j\cdot H)=0$. Both separation inequalities would then read $0>c_1\cdot0=0$ or $0<c_2\cdot0=0$, which are impossible. Hence $\operatorname{Tr}(\rho\Pi_j)>0$ for every $\rho\in\overline{O_1}\cup\overline{O_2}$; since this set is compact and $\rho\mapsto\operatorname{Tr}(\rho\Pi_j)$ is continuous, $\delta:=\min_{\overline{O_1}\cup\overline{O_2}}\operatorname{Tr}(\rho\Pi_j)>0$.

\medskip
\noindent\textbf{Sharp-measurement separation.}
Let $\phi(\rho)=\operatorname{Tr}(\rho H)$. By replacing $H$ with $H+\alpha\Pi_j$ if necessary, we may assume $c_2>0$ (preserving $H=\Pi_jH\Pi_j$ and the gap $c_1-c_2>0$).

Define $A:=f_j^{(1)}(\overline{O_1})$ and $B:=\overline{O_2}\cup f_j^{(1)}(\overline{O_2})$; both are compact, being continuous images of, or unions of continuous images of, the compact sets $\overline{O_1},\overline{O_2}$.

For $\sigma=f_j^{(1)}(\rho_1)\in A$, $\rho_1\in\overline{O_1}$: $\sigma=\Pi_j\rho_1\Pi_j/\operatorname{Tr}(\rho_1\Pi_j)$, and since the hypothesis $\operatorname{Tr}(\rho_1H)>c_1\operatorname{Tr}(\rho_1\Pi_j)$ is strict and $\operatorname{Tr}(\rho_1\Pi_j)>0$ (shown above), dividing gives $\phi(\sigma)=\operatorname{Tr}(\rho_1H)/\operatorname{Tr}(\rho_1\Pi_j)>c_1$ strictly.

For $\sigma=f_j^{(1)}(\rho_2)$, $\rho_2\in\overline{O_2}$: likewise $\phi(\sigma)=\operatorname{Tr}(\rho_2H)/\operatorname{Tr}(\rho_2\Pi_j)<c_2$ strictly.

For $\sigma\in\overline{O_2}$: the hypothesis $\operatorname{Tr}(\sigma H)<c_2\operatorname{Tr}(\sigma\Pi_j)$ is strict, and $\operatorname{Tr}(\sigma\Pi_j)\le1$ with $c_2>0$, so $\phi(\sigma)=\operatorname{Tr}(\sigma H)<c_2\operatorname{Tr}(\sigma\Pi_j)\le c_2$, giving $\phi(\sigma)<c_2$ strictly.

Thus $\phi(\sigma)<c_2$ for every $\sigma\in B$; since $B$ is compact and $\phi$ continuous, the maximum $\max_B\phi$ is attained, so $\max_B\phi<c_2$ strictly. By linearity, any convex combination of points with $\phi<c_2$ again has $\phi<c_2$, so $\phi(\tau)<c_2$ for all $\tau\in\operatorname{Conv}(B)$ as well, and $\max_{\operatorname{Conv}(B)}\phi\le\max_B\phi<c_2$. Similarly $\phi(\sigma)>c_1$ for every $\sigma\in A$, and since $A$ is compact, $\min_A\phi>c_1$. Hence
\[
\min_A\phi \;>\; c_1 \;>\; c_2 \;>\; \max_{\operatorname{Conv}(B)}\phi,
\]
so $A$ and $\operatorname{Conv}(B)$ are separated by the continuous linear functional $\phi$ with a gap of at least $c_1-c_2>0$, and since both are compact,
\[
\varepsilon:=d(A,\operatorname{Conv}(B))>0.
\]

\medskip
\noindent\textbf{Uniform convergence of $f_j^{(\eta)}$ to $f_j^{(1)}$.}
Since $\delta>0$ (shown above), the denominators in the sharp update are uniformly bounded below on $\overline{O_1}\cup\overline{O_2}$. Writing $E_j^{(\eta)}=\bigl(\eta+\frac{1-\eta}{m}\bigr)\Pi_j+\frac{1-\eta}{m}(I-\Pi_j)$, so $M_j^{(\eta)}=a_\eta\Pi_j+b_\eta(I-\Pi_j)$ with $a_\eta=\bigl(\eta+\frac{1-\eta}{m}\bigr)^{1/2}$, $b_\eta=\bigl(\frac{1-\eta}{m}\bigr)^{1/2}$, we have $\|M_j^{(\eta)}-\Pi_j\|_\infty\to0$ and $\operatorname{Tr}(\rho E_j^{(\eta)})\to\operatorname{Tr}(\rho\Pi_j)$ uniformly on $\overline{O_1}\cup\overline{O_2}$, giving
\[
\sup_{\rho\in\overline{O_1}\cup\overline{O_2}}\|f_j^{(\eta)}(\rho)-f_j^{(1)}(\rho)\|_1\longrightarrow0.
\]
Hence there is $\eta_0<1$ such that for all $\eta\in(\eta_0,1)$, this supremum is $<\varepsilon/3$.

\medskip
\noindent\textbf{Stability of separation for $\eta<1$.}
Define $B_\eta:=\overline{O_2}\cup f_j^{(\eta)}(\overline{O_2})$. Uniform convergence gives $d_H(f_j^{(\eta)}(\overline{O_1}),A)<\varepsilon/3$ and $d_H(B_\eta,B)<\varepsilon/3$, hence (convex hull being continuous in Hausdorff metric on compacta) $d_H(\operatorname{Conv}(B_\eta),\operatorname{Conv}(B))<\varepsilon/3$ after possibly increasing $\eta_0$. Therefore $d(f_j^{(\eta)}(\overline{O_1}),\operatorname{Conv}(B_\eta))\ge\varepsilon/3>0$, and since $O_1'(\eta)\subseteq f_j^{(\eta)}(\overline{O_1})$, $O_2'(\eta)\subseteq\operatorname{Conv}(B_\eta)$, it follows that $O_1'(\eta)\cap O_2'(\eta)=\emptyset$ for all $\eta\in(\eta_0,1)$, i.e.\ $(O_1'(\eta),O_2'(\eta))$ is a double-parcel.
\end{proof}

\noindent{\bf Lemma}~{\bf \ref{lem:mobius}}~(M\"obius map for the qubit fuzzy update)

\begin{proof}
Write $\rho$ in Bloch coordinates, $z=2p-1$. The qubit update of Section~\ref{sec:measurement} gives $z'=(z+\eta)/(1+\eta z)$; substituting $z=2p-1$, $z'=2p'-1$ and simplifying yields $p'=(1+\eta)p/(1-\eta+2\eta p)=F_\eta(p)$, independent of $x,y$. That $F_\eta(0)=0$, $F_\eta(1)=1$, $F_\eta'(p)=(1-\eta^2)/(1-\eta+2\eta p)^2>0$, and $F_\eta(p)-p=2\eta p(1-p)/(1-\eta+2\eta p)>0$ on $(0,1)$ follow by direct computation.

Since $x'=\sqrt{1-\eta^2}\,x/(1+\eta z)$ and $y'=\sqrt{1-\eta^2}\,y/(1+\eta z)$ scale $x,y$ by a common non-zero factor depending only on $z$, the map sends the disk of allowed $(x,y)$ at each fixed $z$ (equivalently, fixed $p$) bijectively onto the corresponding disk at $z'$ (equivalently, $p'=F_\eta(p)$). Hence $f_j^{(\eta)}$ maps the $p$-slab $\{a<p<b\}$ — the union of these disks over $z\in(2a-1,2b-1)$ — bijectively onto $\{F_\eta(a)<p<F_\eta(b)\}$.
\end{proof}

\noindent{\bf Example}~\ref{ex:preserved-no-separation}~(Preservation without separation)~{\bf and}~\ref{ex:failure-no-separation}~(Failure without separation)

\begin{proof}
Both are direct consequences of Lemma~\ref{lem:mobius} and the monotonicity of $F_\eta$, together with the observation that a Hermitian $H=\Pi_jH\Pi_j$ on a rank-one $\Pi_j$ is necessarily a real scalar multiple $h\Pi_j$ of $\Pi_j$, so $\operatorname{Tr}(\rho H)/\operatorname{Tr}(\rho\Pi_j)=h$ identically wherever $\operatorname{Tr}(\rho\Pi_j)\neq0$; no choice of $h$ can satisfy $h>c_1>c_2>h$. The remaining verifications are given in the text of each example.
\end{proof}

\noindent{\bf Theorem}~{\bf \ref{thm:volcontraction_qubit}}~(Volume contraction for a qubit)

\begin{proof}
The condition \(\operatorname{Tr}(\rho\Pi_j)=\frac{1+z}{2}\ge c\) implies \(z\ge 2c-1\) for every Bloch vector
\((x,y,z)\) of a state in \(P\).  Hence the denominator in \eqref{jacob} satisfies
\(1+\eta z \ge 1+\eta(2c-1)\).

If \(c\ge\frac12\) then \(2c-1\ge0\) and therefore \(1+\eta z\ge 1\).  Consequently
\(J^{(\eta)}(z)\le(1-\eta^2)^2<1\) for all \(\eta\in(0,1)\) and all \(\rho\in P\).
Integration over \(P\) yields the desired contraction.

If \(0<c<\frac12\) we use the uniform bound
\[
J^{(\eta)}(z) \le \frac{(1-\eta^2)^2}{(1+\eta(2c-1))^4}=:K(\eta).
\]
A direct computation shows that \(K(\eta)<1\) precisely when \(\eta > \dfrac{2(1-2c)}{1+(1-2c)^2}\).
Because \(\eta\mapsto K(\eta)\) is continuous and tends to \(0\) as \(\eta\to1^-\),
there exists \(\eta_0<1\) such that \(K(\eta)<1\) for all \(\eta\in(\eta_0,1)\).
For those \(\eta\) we have \(J^{(\eta)}(z)<1\) uniformly on \(P\), and integration gives
\(\operatorname{Vol}_{\mathrm{HS}}(f_j^{(\eta)}(P))<\operatorname{Vol}_{\mathrm{HS}}(P)\).
\end{proof}

\medskip
\noindent
The general finite-dimensional case (Theorem~\ref{thm:volcontraction_nqubits} below) uses the Jacobian of a projective Kraus update, which we record as two auxiliary results. These are not stated in the body of the paper, so — to avoid disturbing the numbering of the Lemmas and Corollaries stated there — we label them Lemma~A and Corollary~A rather than folding them into the main numbering.

\begin{lemmaA}[Jacobian of a projective linear map]
\label{lem:projective-jacobian}
Let $V$ be a real vector space of dimension $D$, $L:V\to V$ an invertible linear map, and $\ell:V\to\mathbb R$ a linear functional with $\ell(L(x))\neq0$ on the affine section $\{x:\ell(x)=1\}$. The induced map $f(x)=L(x)/\ell(L(x))$ on this $(D-1)$-dimensional affine section has Jacobian
\[
|J_f(x)| = \frac{|\det_{\mathbb R}L|}{|\ell(L(x))|^D}.
\]
\end{lemmaA}

\begin{proof}
Choose coordinates $x=t(1,y)$, $y\in\mathbb R^{D-1}$, so that $\ell(x)=t$ picks out the first coordinate. Write $L(1,y)=(p(y),q(y))$ with $p(y)=p_0+p_1\cdot y$ and $q(y)=q_0+Qy$ affine-linear in $y$ (since $L$ is linear in $(1,y)$). In these coordinates $f$ corresponds to $\psi(y)=q(y)/p(y)$, and
\[
D\psi(y) = \frac{Q\,p(y)-q(y)\,p_1^T}{p(y)^2}.
\]
Suppose first that $Q$ is invertible. By the matrix determinant lemma, $\det(Qp-qp_1^T)=p^{D-2}\det(Q)\bigl[p-p_1^TQ^{-1}q\bigr]$, and the Schur complement formula applied to the block matrix $L=\begin{psmallmatrix}p_0&p_1\\q_0&Q\end{psmallmatrix}$ gives $p(y)-p_1^TQ^{-1}q(y)=\det(L)/\det(Q)$ for every $y$ (the $y$-dependence cancels). Hence
\[
\det(D\psi)(y)=\frac{p(y)^{D-2}\det(L)}{p(y)^{2(D-1)}}=\frac{\det(L)}{p(y)^D}.
\]
For general invertible $L$, both sides of $\det(Qp(y)-q(y)p_1^T)=p(y)^{D-2}\det(L)$ are polynomial in the entries of $L$ (for fixed $y$), and this identity has just been established on the dense open set of $L$ for which $Q$ is invertible — the complement is the zero locus of the polynomial $\det Q$, a proper algebraic subvariety. By continuity, the identity, and hence $\det(D\psi)(y)=\det(L)/p(y)^D$, extends to all invertible $L$. Taking absolute values gives the stated formula, since $\ell(L(x))=p(y)$.
\end{proof}

\begin{corollaryA}[Jacobian of a projective Kraus update]
\label{cor:kraus-jacobian}
Apply Lemma~A with $V=$ the $D=d^2$-dimensional real space of Hermitian $d\times d$ matrices, $L(Y)=MYM^\dagger$, and $\ell=\operatorname{Tr}$. Then, for $M=M_j^{(\eta)}\ge0$,
\[
J^{(\eta)}(\rho) = \frac{(\det M_j^{(\eta)})^{2d}}{\operatorname{Tr}(\rho E_j^{(\eta)})^{d^2}}.
\]
\end{corollaryA}

\begin{proof}
Diagonalizing $M$ with real eigenvalues $\mu_1,\dots,\mu_d\ge0$, and taking as a basis of $V$ the $d$ diagonal directions $E_{ii}$ together with, for each $i<j$, the real and imaginary parts of the off-diagonal entry: $L$ scales each $E_{ii}$ by $\mu_i^2$ and each off-diagonal pair's two-real-dimensional block by $\mu_i\mu_j$ in both directions (determinant $(\mu_i\mu_j)^2$ per pair), so
\[
\det_{\mathbb R}L = \prod_i\mu_i^2\cdot\prod_{i<j}(\mu_i\mu_j)^2 = (\det M)^2\cdot(\det M)^{2(d-1)} = (\det M)^{2d},
\]
using $\prod_{i<j}\mu_i\mu_j=\prod_i\mu_i^{d-1}=(\det M)^{d-1}$ (each eigenvalue occurs in exactly $d-1$ pairs). Since $M=M_j^{(\eta)}\ge0$, $\det M\ge0$ and all quantities are positive, so the absolute values in Lemma~A may be dropped, giving the stated formula by direct substitution, since $\ell(L(\rho))=\operatorname{Tr}(M\rho M^\dagger)=\operatorname{Tr}(\rho M^\dagger M)=\operatorname{Tr}(\rho E_j^{(\eta)})$.
\end{proof}

\noindent{\bf Theorem}~{\bf \ref{thm:volcontraction_nqubits}}~(Volume contraction, general finite-dimensional case)

\begin{proof}
By Corollary~A,
\[
J^{(\eta)}(\rho)=\frac{(\det M_j^{(\eta)})^{2d}}{\operatorname{Tr}(\rho E_j^{(\eta)})^{d^2}}.
\]
The eigenvalues of \(E_j^{(\eta)}\) are \(a_\eta:=\eta+\frac{1-\eta}{m}\) (on the range of \(\Pi_j\),
multiplicity \(r\)) and \(b_\eta:=\frac{1-\eta}{m}\) (on the orthogonal complement, multiplicity
\(d-r\)). Since \(M_j^{(\eta)}=\sqrt{E_j^{(\eta)}}\) (positive square root), its eigenvalues
are \(\sqrt{a_\eta}\) (multiplicity \(r\)) and \(\sqrt{b_\eta}\) (multiplicity \(d-r\)), so
\[
\det M_j^{(\eta)} = a_\eta^{r/2}\,b_\eta^{(d-r)/2},
\]
and consequently
\[
\bigl(\det M_j^{(\eta)}\bigr)^{2d}
      = a_\eta^{rd}\,b_\eta^{(d-r)d}
      = \Bigl(\eta+\frac{1-\eta}{m}\Bigr)^{rd}
        \Bigl(\frac{1-\eta}{m}\Bigr)^{(d-r)d}.
\]

Using the assumption \(\operatorname{Tr}(\rho\Pi_j)\ge c\) we obtain
\[
\operatorname{Tr}(\rho E_j^{(\eta)}) = \eta\operatorname{Tr}(\rho\Pi_j)+\frac{1-\eta}{m}
                       \ge \eta c + \frac{1-\eta}{m}.
\]

Define
\[
\alpha(\eta)=\Bigl(\eta+\frac{1-\eta}{m}\Bigr)^{rd}
              \Bigl(\frac{1-\eta}{m}\Bigr)^{(d-r)d},
\qquad
\beta(\eta)=\Bigl(\eta c + \frac{1-\eta}{m}\Bigr)^{d^2}.
\]
Then $J^{(\eta)}(\rho) \le \alpha(\eta)/\beta(\eta)$ for all $\rho\in P$.

As \(\eta\to1^-\): since \(d-r>0\), \(\bigl(\frac{1-\eta}{m}\bigr)^{(d-r)d}\to0\) while \(\eta+\frac{1-\eta}{m}\to1\), so \(\alpha(\eta)\to0\); and \(\eta c+\frac{1-\eta}{m}\to c>0\), so \(\beta(\eta)\to c^{d^2}>0\). Hence \(\alpha(\eta)/\beta(\eta)\to0\), and by continuity there exists \(\eta_0<1\) such that \(J^{(\eta)}(\rho)<1\) uniformly on \(P\) for all \(\eta\in(\eta_0,1)\).

Since $f_j^{(\eta)}$ is a homeomorphism (Proposition~\ref{prop:singleupdate}), the change-of-variables formula gives
\[
\operatorname{Vol}_{\mathrm{HS}}\!\left(f_j^{(\eta)}(P)\right)
= \int_P J^{(\eta)}(\rho) \, \mathrm{d}\rho
< \int_P 1 \, \mathrm{d}\rho
= \operatorname{Vol}_{\mathrm{HS}}(P).
\]
\end{proof}

\noindent{\bf Corollary}~{\bf \ref{cor:strict-increase}}~(Strict increase of geometric information)

\begin{proof}
Since $O_2\neq\varnothing$ is open, $\operatorname{Vol}_{\mathrm{HS}}(O_2)>0$. By construction $O_2\subseteq O_2'(\eta)$, so $\operatorname{Vol}_{\mathrm{HS}}(O_2'(\eta))\ge\operatorname{Vol}_{\mathrm{HS}}(O_2)>0$. By Theorem~\ref{thm:volcontraction_qubit} or Theorem~\ref{thm:volcontraction_nqubits}, $\operatorname{Vol}_{\mathrm{HS}}(O_1'(\eta))<\operatorname{Vol}_{\mathrm{HS}}(O_1)$. Therefore
\[
I_{\mathrm{HS}}(O_1'(\eta),O_2'(\eta))
=\frac{\operatorname{Vol}_{\mathrm{HS}}(O_2'(\eta))}{\operatorname{Vol}_{\mathrm{HS}}(O_1'(\eta))}
>\frac{\operatorname{Vol}_{\mathrm{HS}}(O_2)}{\operatorname{Vol}_{\mathrm{HS}}(O_1)}
=I_{\mathrm{HS}}(O_1,O_2).
\]
\end{proof}

\subsection*{Proofs from Section~\ref{sec:entanglement} (Entanglement and Bell Violation in IQM)}

\noindent{\bf Theorem}~{\bf \ref{bell-violation}}~(Finite-precision Bell violation)

\begin{proof}
Since $\rho\mapsto\operatorname{Tr}(\rho S)$ is continuous and equals $2\sqrt2>2$ at $\ket{\Phi^+}\bra{\Phi^+}$, the set $\{\rho:\operatorname{Tr}(\rho S)>2\}$ is an open neighbourhood of the Bell state; since experimental parcels form a basis for the state-space topology (Section~\ref{sec:parcels}), some experimental parcel $P$ contains the Bell state and lies inside this neighbourhood. Tsirelson's bound $\operatorname{Tr}(\rho S)\le2\sqrt2$ for every $\rho$ then gives the stated range.
\end{proof}